\documentclass[a4paper, 12pt, reqno, twoside]{amsart}
\usepackage{amsmath, amsthm, amssymb}
\usepackage{fullpage, enumerate, color, tikz, dsfont, subcaption, booktabs, array}
\usetikzlibrary{automata, positioning, arrows.meta, decorations.pathmorphing}
\usepackage[colorlinks=true,urlcolor=cyan,linkcolor=blue,citecolor=magenta]{hyperref}

\theoremstyle{plain}
\newtheorem{theorem}{Theorem}
\newtheorem{proposition}[theorem]{Proposition}
\newtheorem{lemma}[theorem]{Lemma}
\newtheorem{corollary}[theorem]{Corollary}

\theoremstyle{definition}
\newtheorem{definition}[theorem]{Definition}
\newtheorem{remark}[theorem]{Remark}

\newcommand{\set}[2]{\left\{#1\ \middle\vert\ #2\right\}}
\newcommand{\leb}{\mathrm{LEB}}
\newcommand{\gleb}{\mathrm{G_\leb}}

\newcommand{\C}{\mathbb{C}}

\newcommand{\R}{\mathbb{R}}
\newcommand{\N}{\mathbb{N}}
\newcommand{\Z}{\mathbb{Z}}

\newcommand{\HH}{\mathbb{H}}

\newcommand{\absolute}[1] {\left|{#1}\right|}

\newcommand{\norm}[1]{\left\|{#1}\right\|}
\newcommand{\inpro}[2]{\langle{#1},{#2}\rangle}

\newcommand{\we}{w_{\text{even}}}
\newcommand{\wo}{w_{\text{odd}}}
\newcommand{\mRe}{\mathrm{Re}}
\newcommand{\mIm}{\mathrm{Im}}
\newcommand{\mnull}{\mathrm{N}}
\newcommand{\I}{\text{I}}
\newcommand{\II}{\text{II}}
\newcommand{\III}{\text{III}}
\newcommand{\IV}{\text{IV}}
\newcommand{\V}{\text{V}}
\newcommand{\VI}{\text{VI}}

\newcommand{\Rd}{R_{\IV \cup \V \cup \VI}}

\newcommand{\cA}{\mathcal{A}}
\newcommand{\cD}{\mathcal{D}}

\title{On the stability, complexity, and distribution of similarity classes of the longest edge bisection process for triangles}

\author{Daniel Kalmanovich}
\address{Department of Computer Science, Ariel University, Israel} 
\email{danielk@ariel.ac.il}
\urladdr{https://danielkalmanovich.wixsite.com/website}

\author{Yaar Solomon}
\address{Department of Mathematics, Ben-Gurion University of the Negev, Israel} 
\email{yaars@bgu.ac.il}
\urladdr{https://www.math.bgu.ac.il/~yaars/}

\begin{document}

\begin{abstract}
    The Longest Edge Bisection of a triangle is performed by joining the midpoint of its longest edge to the opposite vertex. Applying this procedure iteratively produces an infinite family of triangles. Surprisingly, a classical result of Stynes (1980) shows that for any initial triangle, the elements of this infinite family fall into finitely many similarity classes.

    While the set of classes is finite, it turns out that a far smaller, periodic subset of ``fat'' triangles effectively dominates the final mesh structure. This subset is comprised of periodic orbits of length four, which we refer to as {\bf terminal quadruples}. We prove the following asymptotic area distribution result: for every initial triangle, the portion of area occupied by these terminal quadruples tends to one, with the convergence occurring at an exponential rate. In fact, we provide the precise distribution of triangles in every step. We introduce the {\bf bisection graph} and use spectral methods to prove this result.

    Given this dominance, we provide a complete characterization of triangles possessing a single terminal quadruple, while conversely exhibiting a sequence of triangles with an unbounded number of terminal quadruples. Furthermore, we reveal several fundamental geometric properties of the points of a terminal quadruple, laying the groundwork for studying the geometric distribution of the entire orbit.
\end{abstract}

\maketitle

\section{Introduction}

\subsection{Background}

The Longest Edge Bisection ($\leb$) process is a straightforward iterative procedure, formally defined and studied by Rosenberg and Stenger in 1975~\cite{RoseS75}. In each step, every existing triangle is bisected by the median to its own longest edge, splitting it into two daughter triangles. Rosenberg and Stenger proved the critical property of non-degeneracy: for an initial triangle $\Delta_0$ with smallest angle $\alpha_0$, the angles of all triangles generated by the infinite process are bounded from below by $\frac{\alpha_0}{2}$. Another requirement for a useful triangulation process is the convergence of the mesh diameter to zero. Kearfott~\cite{Kear78} first proved this for $\leb$ in arbitrary dimensions; later, Stynes~\cite{Styn79, Styn80} and Adler~\cite{Adle83} provided sharpened bounds on the convergence rates.

A landmark result, known as the {\bf Finite Similarity Classes (FSC) theorem}, was established by Stynes~\cite[Corollary 1]{Styn80}. This theorem states that, although iteratively bisecting a single initial triangle produces infinitely many triangles, they belong to only finitely many similarity classes. This result was later reproved by Adler~\cite{Adle83} using an elegant inductive argument, and subsequently by Perdomo and Plaza~\cite{PerdP14}, who recast the analysis within the framework of hyperbolic geometry, a setting that proved to be particularly natural.

In their paper, Rosenberg and Stenger identified a specific region where triangles enter a periodic orbit of length $4$ under the $\leb$ process~\cite[Lemma 4]{RoseS75}, forming what we call a terminal quadruple. Building on this work and the FSC theorem, Stynes~\cite[Corollary 2]{Styn80} demonstrated that the proportion of the area occupied by triangles in these terminal quadruples converges to $1$.

These results pave the way for a deeper, quantitative analysis of the $\leb$ process as a dynamical system. While it is known that the variety of shapes is finite and that terminal quadruples asymptotically exhaust the area of the initial triangle, this does not describe the rate of convergence or the precise distribution of triangle shapes. Partial results in these directions were previously noted in~\cite{GutiGR07}.

\subsection{Preliminaries and main results}\label{subsec:prelim}

To analyze the $\leb$ process as a dynamical system, we employ the framework introduced by Perdomo and Plaza~\cite{PerdP14}, which utilizes the machinery of hyperbolic dynamics in a moduli space of triangle shapes modeled in the hyperbolic upper half-plane $\HH$. Given an arbitrary triangle, we normalize it by scaling its longest edge to length 1 and mapping this edge to the interval $[0, 1]$ on the real axis. By placing the third vertex in the upper half-plane $\HH$ such that the shortest edge lies on the left, we bijectively identify every similarity class of triangles with a unique point $z$ in the domain $\cD\subseteq\HH$ defined by (see Figure \ref{fig:normalized_region}): 
\[
\cD = \set{z\in \C}{\mRe(z)\le\frac12,\:  \mIm(z)> 0,\: \absolute{z-1}\le 1}.
\]
In light of this bijection, we will refer to elements $z \in \cD$ interchangeably as points or triangles. For such $z \in \cD$, bisecting the triangle by the median to its longest edge defines two maps $L, R: \cD \to \cD$, where $L(z)$ is the {\bf left triangle} and $R(z)$ is the {\bf right triangle}. As explained in~\cite{PerdP14}, the domain $\cD$ is subdivided into $6$ subregions by the geodesics $\mRe(z)=\frac14$, $\absolute{z} = \frac12$ and $\absolute{z-\frac12}=\frac12$ (see~Figure \ref{fig:normalized_subdivision}), where the functions $L$ and $R$ are defined in a piecewise manner, by a M\"obius or anti-M\"obius transformation on each one of these subregions. We consider the action of the semigroup $\inpro{L}{R}$ on $\cD$ and refer to words in this semigroup as the functions that they correspond to. We also adopt standard terminology from dynamical systems, such as the orbit of a point. The explicit formulas and geometric interpretations of these transformations are summarized in Table~\ref{tab:L_and_R}. An illustrative application presenting the geometry of the various transformations is available at~\url{https://danielkalmanovich.wixsite.com/website/leb}.
The reader is referred to \cite[\S 3]{Ahlfors} for more details on the action by M\"obius transformations on the upper half-plane.

\begin{figure}[h!]
    \centering
    \begin{minipage}[b]{.31\textwidth}
    \begin{tikzpicture}[scale=5.5] 
        
        \coordinate (A) at (0,0);
        \coordinate (B) at (0.5,0);
        \coordinate (C) at (0.5, {sqrt(3)/2});
        
        \filldraw[fill=green!20, draw=green!60!black, line width=0.8pt, line join=round]
          (A) -- (B) -- (C) arc (120:180:1) -- cycle;
        
        \draw[->, line width=0.5pt] (-0.2,0) -- (1.2,0) node[right] {$\mRe$};
        \draw[->, line width=0.5pt] (0,-0.2) -- (0,1.2) node[above] {$\mIm$};
        
        \foreach \x in {0.5, 1}
          \draw (\x, 1.5pt) -- (\x, -1.5pt) node[below=2pt] {\small$\x$};
        
        \foreach \y in {0.5, 1}
          \draw (1.5pt, \y) -- (-1.5pt, \y) node[left=2pt] {\small$\y i$};
          
        \node[below left=2pt] at (0,0) {\small$0$};
        
    \end{tikzpicture}
    \subcaption{The normalized region $\cD$}
    \label{fig:normalized_region}
    \end{minipage}
    \hfill
    \begin{minipage}[b]{.45\textwidth}
    \begin{tikzpicture}[scale=8.6],
        \coordinate (A) at (0,0.5);
        \coordinate (B) at (0.5,0.5);
        \coordinate (C) at (0.5, {0.5+sqrt(3)/2}); 
        
        \fill[green!15] (A) -- (B) -- (C) arc (120:180:1) -- cycle;
        
        
        \coordinate (D) at (0.25, 0.5);
        \coordinate (E) at (0.25, {0.5+sqrt(7)/4}); 
        \draw[blue!70!black, line width=0.6pt] (D) -- (E);
        
        \draw[blue!70!black, line width=0.6pt] (B) arc (0:{acos(0.25)}:0.5);
        
        \draw[blue!70!black, line width=0.6pt] (A) arc (180:90:0.5);
        
        
        \node at (0.375, 1.1) {I};
        \node at (0.21, 1.025) {II};
        \node at (0.15, 0.925) {III};
        \node at (0.41, 0.91) {IV};
        \node at (0.125, 0.6) {V};
        \node at (0.375, 0.6) {VI};

        \node[blue!70!black] at (0.17, 1.22) {$\mRe(z)=\frac14$};
        \node[blue!70!black] at (0.04, 1) {$\absolute{z}=\frac12$};
        \node[blue!70!black] at (0.625, 1) {$\absolute{z-\frac12}=\frac12$};
        
        \draw[green!60!black, line width=0.8pt, line join=round]
          (A) -- (B) -- (C) arc (120:180:1) -- cycle;
    \end{tikzpicture}
    \subcaption{The subdivision of $\cD$ into $6$ subregions}
    \label{fig:normalized_subdivision}
    \end{minipage}
\caption{The normalized region $\cD$ and its subdivision.}
\end{figure}
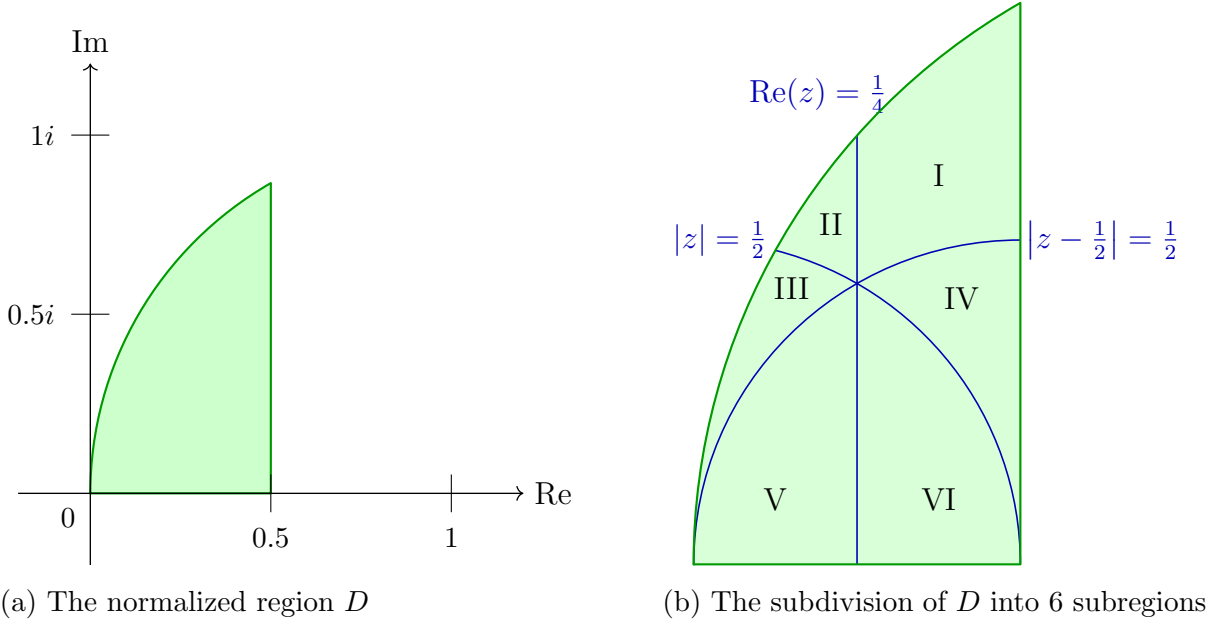

{\small \begin{table}[h!]
\begin{tabular}{c l >{\raggedright\arraybackslash}p{10cm}}
\toprule
\textbf{Region} & \textbf{Function} & \textbf{Geometric interpretation in $\cD$} \\
\midrule
I   & $L(z) = \frac{1}{2\overline{z}}$ &
An inversion with respect to the circle $|z|=\frac{\sqrt2}{2}$. 
\\
\addlinespace
II  & $L(z) = \frac{-1}{2z-1}$      & 
An elliptic M\"{o}bius map with a fixed point at $\frac{1}{4} + \frac{\sqrt(7)}{4} i\in \cD$. Points rotate counterclockwise by an angle of $\arccos\left(-\frac{3}{4}\right)$ about the fixed point. 
\\
\addlinespace
III & $L(z) = \frac{2\bar{z}}{2\bar{z}-1}$ & 
A hyperbolic anti-M\"{o}bius map. Points spiral from the repeller at $z=0$ to the attractor at $z=\frac32$. 
\\
\addlinespace
IV  & $L(z) = \frac{2z-1}{2z}$        & 
An elliptic M\"{o}bius map with a fixed point at $\frac12 + \frac12 i\in \cD$. Points rotate clockwise by an angle of $\frac{\pi}{2}$ about the fixed point.
\\
\addlinespace
V  & $L(z) = 2z$               & 
A hyperbolic M\"{o}bius map that dilates by $2$.
\\
\addlinespace
VI   & $L(z) = 1 - 2\bar{z}$      & 
A hyperbolic anti-M\"{o}bius map. A reflection across $\mRe(z)=\frac13$ combined with a dilation by a factor of $2$.
\\
\addlinespace
\midrule
\addlinespace
$\I\cup\II\cup\III$ & $R(z) = \frac{-1}{2z-2}$ & 
An elliptic M\"{o}bius map with a fixed point at $\frac12 + \frac12 i\in \cD$. Points rotate counterclockwise by an angle of $\frac{\pi}{2}$ about the fixed point.
\\
\addlinespace
$\IV\cup\V\cup\VI$ & $R(z) = \frac{2\bar{z}-1}{2\bar{z}-2}$ & 
An inversion with respect to the circle $|z-1|=\frac{\sqrt2}{2}$. \\
\bottomrule
\end{tabular}
\caption{The piecewise functions $L(z)$ and $R(z)$.}
\label{tab:L_and_R}
\end{table}}

In view of the fact that the $\leb$ process eventually produces predominantly terminal quadruples, it seems beneficial to move beyond this qualitative result and achieve a comprehensive quantitative and geometric understanding of the process. We are led to the following questions:

\begin{enumerate}
\item 
How many distinct terminal quadruples can arise from a given initial triangle $z$?
\item 
What proportion of the total area of the initial triangle $z$ do terminal quadruple triangles occupy after $k$ bisection steps?
\item 
Which geometric properties does the orbit of a triangle possess?
\end{enumerate}

To answer these questions, we introduce the {\bf bisection graph} of a triangle (see \S\ref{section:graph}) and analyze it using spectral methods. Furthermore, we adopt and refine the methods used in~\cite{PerdP14} to study the geometric properties of the orbit of $z$ under iterative $\leb$. We denote by:
\begin{itemize}
    \item 
    $\leb(z)$ the orbit of $z$,
    \item
    $q(z)$ the number of terminal quadruples that arise from $z$,
    \item
    $l(z)=\absolute{\leb(z)}$. 
\end{itemize}  

Our main results are the following.

\begin{theorem}\label{thm:occupying_area}
     For a triangle $z$ and $j\in\N$, let $w_j\in \R^{l(z)}$ be the probability vector describing the partition of the area of $z$ into triangles of different similarity classes, after $j$ bisection steps. Then for every triangle $z$ we have:
    \begin{enumerate}
        \item 
        The limits 
        $w_{even} = \displaystyle{\lim_{j\to\infty}w_{2j}}$  and 
        $w_{odd} = \displaystyle{\lim_{j\to\infty}w_{2j+1}}$ exist, and their non-zero entries correspond to elements in terminal quadruples.
        \item 
        There exists $\xi\in (0,1)$ such that for every $j\in \N$, the area of $z$ that is occupied by triangles not belonging to terminal quadruples after $j$ steps is $O(\xi^j)$.
    \end{enumerate}
\end{theorem}

This theorem immediately implies Stynes' result, stating that the area of any initial triangle $z$ gets filled with triangles that belong to terminal quadruples (see~\cite[Corollary 2]{Styn80}). Furthermore, Theorem~\ref{thm:occupying_area} provides the precise distribution vector at every step and the rate of convergence. 

It follows that the number of triangle shapes that one sees is essentially $4q(z)$ rather than $l(z)$. For instance, taking an acute isosceles triangle with apex angle $\alpha$ we have $l(z)=\Theta\left(\log\left(\frac{1}{\alpha}\right)\right)$ (see ~\cite[Lemma~2]{GutiGR07} for the upper bound, and the lower bound is obtained by inspecting the function $L$ in region $\III$), while according to Theorem~\ref{thm:q(z)=1} below we have $q(z)=1$. In general, the following theorems shed light on the quantity $q(z)$, for $z$ in the different regions.  

\begin{theorem}\label{thm:q(z)=1}
    For every $z\in \I\cup \II\cup\III\cup\IV$ we have $q(z)=1$. 
\end{theorem}

\begin{theorem}\label{thm:arbitrarily_many_quadruples}
    For every $n\in\N$, there are triangles $z$ in regions $\V$ and $\VI$ for which $q(z) > n$. 
\end{theorem}

The rest of the paper is structured as follows: in \S\ref{section:tq} we provide a classification of the non-generic triangles, a short discussion of geometric properties of terminal quadruples, a proof of Theorem~\ref{thm:q(z)=1}, and a proof of Theorem~\ref{thm:arbitrarily_many_quadruples}. In \S\ref{section:graph} we introduce the bisection graph and prove our main result, Theorem~\ref{thm:occupying_area}, which motivates our study of the terminal quadruples in \S\ref{section:tq}. Concluding remarks and open problems are presented in~\S\ref{sec:concluding}. 

\section{Terminal quadruples}\label{section:tq}

We start this section with the definition of the generic triangles with respect to the $\leb$ process and their classification.

\subsection{Generic triangles and triangles with orbits of size less than $4$}\label{subsec:non_gen_terminal}

It can happen that $L(z)$ is similar to $R(z)$, or that one of them is similar to $z$, thus we define:

\begin{definition}\label{def:generic}
    A triangle $z\in \cD$ is {\bf generic} if for each $w\in\leb(z)$, the points $w, L(w)$ and $R(w)$ are all distinct.    
\end{definition}  

The following proposition is a characterization of the non-generic triangles (see Figure~\ref{fig:non-generic}).

\begin{proposition}\label{prop:non-generic}
    A triangle $z\in \cD$ is generic if and only if  $\leb(z)$ does not intersect the following $3$ geodesics: 
    \begin{equation}\label{eq:non-generic}
    \mRe(w)=\frac{1}{2},~~~ 
    \absolute{w}=\frac{\sqrt2}{2},~~~    
    \absolute{w-1}=\frac{\sqrt2}{2}.
    \end{equation}
\end{proposition}

\begin{proof}
    To help the reader follow the computations in the proof, we use subscripts to indicate which piece of the piecewise maps $L$ and $R$ is being applied. 

    First, we address points $w\in \cD$ that are fixed under either $L$ or $R$. Since $L_\I$ is an inversion with respect to the geodesic $\absolute{w}=\frac{\sqrt2}{2}$, we have $L_\I(w)=w$ for points lying on this geodesic. Similarly, $R_{\IV\cup\V\cup\VI}$ is an inversion with respect to the geodesic $\absolute{w-1}=\frac{\sqrt2}{2}$ and so $R_{\IV\cup\V\cup\VI}(w)=w$ there. Each one of $L_\II, L_\IV$ and $R_{\I\cup\II\cup\III}$ has a unique fixed point, which is on one of these two geodesics, and $L_\III, L_\V$ and $L_\VI$ have no fixed points in $\cD$. Thus $w=R(w)$ or $w=L(w)$ if and only if $\absolute{w-1}=\frac{\sqrt2}{2}$ or $\absolute{w}=\frac{\sqrt2}{2}$.  

    Next, we show that $L(w)=R(w)$ if and only if $\mRe(w)=\frac12$. Notice that for every $w\in \cD$, the triangle $R(w)$ is not acute. In case $L(w)$ is obtuse, one can easily see that $L(w)\neq R(w)$ because the obtuse angle of $L(w)$ is strictly smaller than that of $R(w)$. So we may assume that $L(w)$ is not obtuse. Thus $L(w)=R(w)$ if and only if $L(w)$ and $R(w)$ are the same right-angled triangle. Observe that if $L(w)=R(w)$ then $w$ is isosceles, and if the apex of an isosceles triangle $w$ is $<\frac{\pi}{3}$, then its longest edge is its leg, thus $L(w)\neq R(w)$ in that case. Hence, we conclude that $L(w)=R(w)$ if and only if $w$ is an isosceles triangle with apex $\ge\frac{\pi}{3}$. Thus $L(w)=R(w)$ if and only if $\mRe(w)=\frac12$.
\end{proof}

\begin{figure}[h!]
\centering
\begin{minipage}[b]{.475\textwidth}

\begin{tikzpicture}[scale=9] 

    \coordinate (A) at (0,0);
    \coordinate (B) at (0.5,0);
    \coordinate (C) at (0.5, {sqrt(3)/2});

    \fill[green!15] (A) -- (B) -- (C) arc (120:180:1) -- cycle;

    \draw[black!30, line width=0.6pt, opacity=0.6]
        (0.25, 0) -- (0.25, {sqrt(7)/4}) 
        (0.5,0) arc (0:{acos(0.25)}:0.5) 
        (0,0) arc (180:90:0.5); 
     
    \draw[black!30, line width=0.6pt, line join=round]
      (A) -- (B) -- (C) arc (120:180:1) -- cycle;

    
    \draw[green!60!black, line width=2pt] (0.5, 0.001) -- (0.5, {sqrt(3)/2});

    \node[green!60!black, text width=2cm] at (0.65, 0.35) {$\mRe(z)=\frac12$};

    \draw[red, line width=2pt] 
    ({acos(sqrt(2)/4)}:{sqrt(2)/2}) arc ({acos(sqrt(2)/4)}:45:{sqrt(2)/2});

    \node[red] at (0.15, 0.7) {$\absolute{z}=\frac{\sqrt{2}}{2}$};


    \draw[red, line width=2pt] 
        (1,0) ++(135:{sqrt(2)/2}) arc (135:179.9:{sqrt(2)/2});
    \node[red] at (0.15,0.1) {$\absolute{z-1}=\frac{\sqrt{2}}{2}$};

    \fill[blue] (0.5, 0.5) circle (0.2pt) node[right, xshift=3pt] {$\zeta=\frac12+\frac12i$};

\end{tikzpicture}
\subcaption{The non-generic triangles in $\cD$: The triangles on the upper {\color{red}red} curve satisfy $z=L(z)$, the triangles on the lower {\color{red}red} curve satisfy $z=R(z)$. The triangles on the {\color{green!60!black}green} curve satisfy $L(z)=R(z)$.}
\label{fig:non-generic}
\end{minipage}
\hfill
\begin{minipage}[b]{.475\textwidth}

\begin{tikzpicture}[scale=9] 

    \coordinate (A) at (0,0);
    \coordinate (B) at (0.5,0);
    \coordinate (C) at (0.5, {sqrt(3)/2});
    \coordinate (F) at (0.25, {sqrt(3)/4});

    \fill[green!15] (A) -- (B) -- (C) arc (120:180:1) -- cycle;

    \fill[blue!15] 
            (C) arc (120:{acos(-0.75)}:1) -- (F) arc (60:{acos(0.75)}:0.5) arc ({acos(1/8)}:60:1/3) -- cycle;

    \draw[black!30, line width=0.6pt, opacity=0.6]
        (0.25, 0) -- (0.25, {sqrt(7)/4}) 
        (0.5,0) arc (0:{acos(0.25)}:0.5) 
        (0,0) arc (180:90:0.5); 
     
    \draw[black!30, line width=0.6pt, line join=round]
      (A) -- (B) -- (C) arc (120:180:1) -- cycle;


    \draw[red, line width=2pt] 
    ({acos(sqrt(2)/4)}:{sqrt(2)/2}) arc ({acos(sqrt(2)/4)}:45:{sqrt(2)/2});

    \node[red] at (0.15, 0.7) {$\absolute{z}=\frac{\sqrt{2}}{2}$};

    \draw[red, line width=2pt] 
        (1,0) ++(135:{sqrt(2)/2}) arc (135:{acos(-5*sqrt(2)/8)}:{sqrt(2)/2});

    \node[red] at (0.35, 0.3) {$\absolute{z-1}=\frac{\sqrt{2}}{2}$};
    
    \draw[green!60!black, line width=2pt] (0.5, 0.5) -- (0.5, {sqrt(3)/2});


    \draw[green!60!black, line width=2pt] 
        (0.5, 0) ++(90:0.5) arc (90:120:0.5);

    \node[green!60!black] at (0.1, 0.45) {$\absolute{z-\frac12}=\frac12$};

    \draw[green!60!black, line width=2pt] 
        (0.5, 0.5) -- (0.5, {sqrt(3)/6});

    \node[green!60!black, text width=2cm] at (0.65, 0.35) {$\mRe(z)=\frac12$
    };
    
    \fill[blue] (0.5, 0.5) circle (0.4pt) node[right, xshift=3pt] {$\zeta=\frac12+\frac12i$};

\end{tikzpicture}
\subcaption{The non-generic triangles in the shaded terminal region $\cA$: the {\color{blue}blue} $\zeta$ generates $1$ triangle. The triangles on the {\color{red}red} curves generate $2$ triangles. The triangles on the {\color{green!60!black}green} curves generate $3$ triangles.}
\label{fig:non-generic_terminal}
\end{minipage}
\caption{The non-generic locus.}
\end{figure}
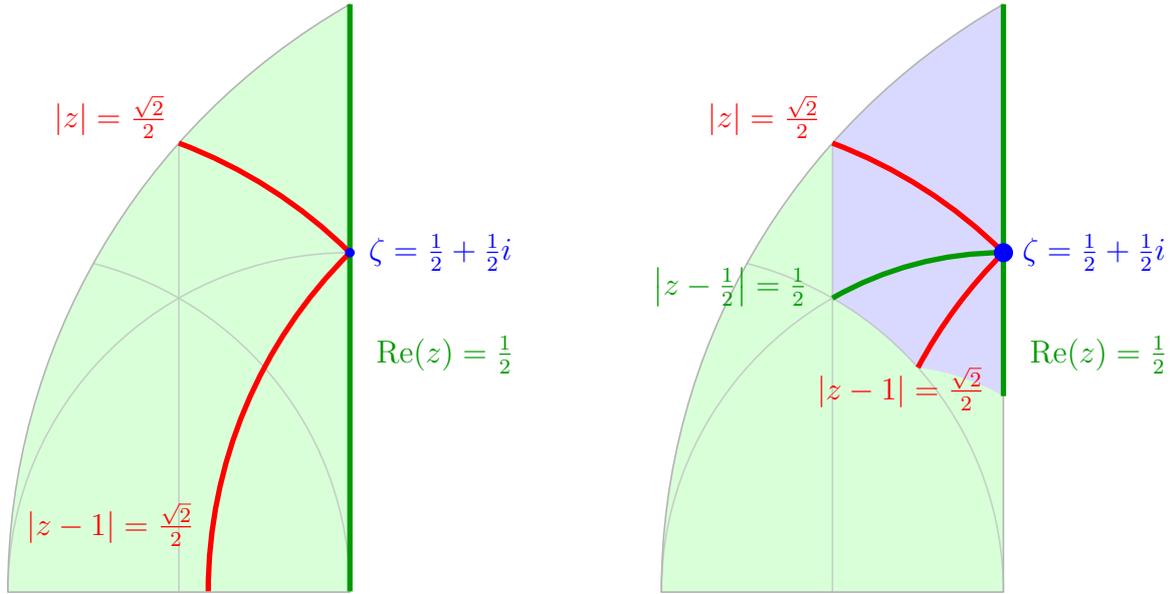

\begin{remark}
    In view of Proposition~\ref{prop:non-generic}, the set of non-generic triangles is contained in a countable union of arcs in $\cD$, i.e., the inverse images of the arcs described in~\eqref{eq:non-generic} with respect to iterative applications of $L$ and $R$, thus has measure zero. In particular, the set of generic triangles is of full measure.
\end{remark}

It is well-known that any region $\I$ triangle generates at most $4$ similarity classes (see e.g. \cite{PerdP14}). Generic triangles produce orbits of $4$ distinct elements, while non-generic triangles produce orbits of $4$ elements, some of which are repeated. Moreover, this orbit exhibits a specific {\bf periodic} structure under the bisection maps, as in the following definition.

\begin{definition}
An orbit $\{z_1, z_2, z_3, z_4\}$ (with not necessarily distinct elements) under the $\leb$ process is called a {\bf terminal quadruple} if $L$ and $R$ map $\{z_1, z_4\}$ to $\{z_2, z_3\}$, and vice versa.
\end{definition}

In fact, the locus of triangles exhibiting this behavior is larger. Since $\I$ is invariant under $L$, it follows that any triangle $z$ in the region 
\[
\cA=\I \cup R(\I) \subseteq \I \cup \IV
\] 
is such a triangle (see~\cite[Lemma~4]{RoseS75}, where this was already established). The fact that $\cA$ is precisely the locus of terminal quadruples is shown in Corollary~\ref{cor:A_is_terminal} below.

The precise characterization of the non-generic terminal quadruple triangles is provided in the following corollary.

\begin{corollary}
\label{cor:triangles_with_l(z)<4}
    Let $z\in \cD$ and set $\zeta=\frac12+\frac12i$, then:
    \begin{enumerate}[(a)]
        \item\label{propitem:q(z)=1} 
        $l(z)=1$ if and only if $z=\zeta$. 
        \item\label{propitem:q(z)=2} 
        $l(z)=2$ if and only if $z\in\set{w\in \cA}{\absolute{w}=\frac{\sqrt{2}}{2} ~~\text{ or }~~ |w-1|=\frac{\sqrt{2}}{2}}\setminus\{\zeta\}$. 
        \item\label{propitem:q(z)=3} 
        $l(z)=3$ if and only if $z\in\set{w\in \cA}{\absolute{w-\frac12}=\frac{1}{2} ~~\text{ or }~~ \mRe(w)=\frac{1}{2}}\setminus\{\zeta\}$.
    \end{enumerate}
\end{corollary}

\begin{proof}
    Since the geodesic $\absolute{w-\frac12}=\frac{1}{2}$ in $\cA$ is the image of the geodesic $\mRe(w)=\frac{1}{2}$ in $\cA$ under $L$ and $R$, this is a direct consequence of Proposition \ref{prop:non-generic} (see Figure \ref{fig:non-generic_terminal}). 
\end{proof}

\subsection{The geometric structure of the terminal quadruples}\label{subsec:geometric_structure}

While running simulations, we have observed many remarkable phenomena regarding the geometric position of the points of $\leb(z)$ in $\cD$. The following result, which describes the geometric distribution of the points of a terminal quadruple, is a forerunner of a complete geometric distribution analysis of $\leb(z)$ (see \S\ref{subsec:leb_geo_dist}).  

\begin{proposition}\label{prop:terminal_sets_are_on_a_circle}
    For every $z_1\in \cA$, the set $\{z_1,z_2,z_3,z_4\}$ that is generated from $z_1$ lies on a hyperbolic circle centered at $\zeta = \frac12+\frac12i$. In particular, every terminal quadruple lies on an Euclidean circle centered at a point on the line $\mRe(z)=\frac12$. 
\end{proposition}

\begin{proof}
    For $z\in \I$, the point $L(z)$, is obtained by an inversion with respect to the geodesic $\set{w}{|w|=\frac{\sqrt{2}}{2}}$, that contains $\zeta=\frac12+\frac12i$. Since $\{z_1,z_2,z_3,z_4\}\cap\I\neq\emptyset$, we may assume that $z_1\in \I$. The points $z_1$ and $L(z_1)$ are both in region $\I$ and have the same hyperbolic distance from $\zeta$. Furthermore, for $z\in \I$, the point $R(z)\in \IV$ is obtained by a hyperbolic rotation of $\frac{\pi}{2}$ around $\zeta$. Thus $R(z_1)$ and $RL(z_1)$ are also at the same hyperbolic distance from $\zeta$ as $z_1$. The fact that a hyperbolic circle is a Euclidean circle is straightforward to verify using the two metrics. 
 \end{proof}

\begin{figure}[h!]
    \centering
    \includegraphics[scale=0.619]{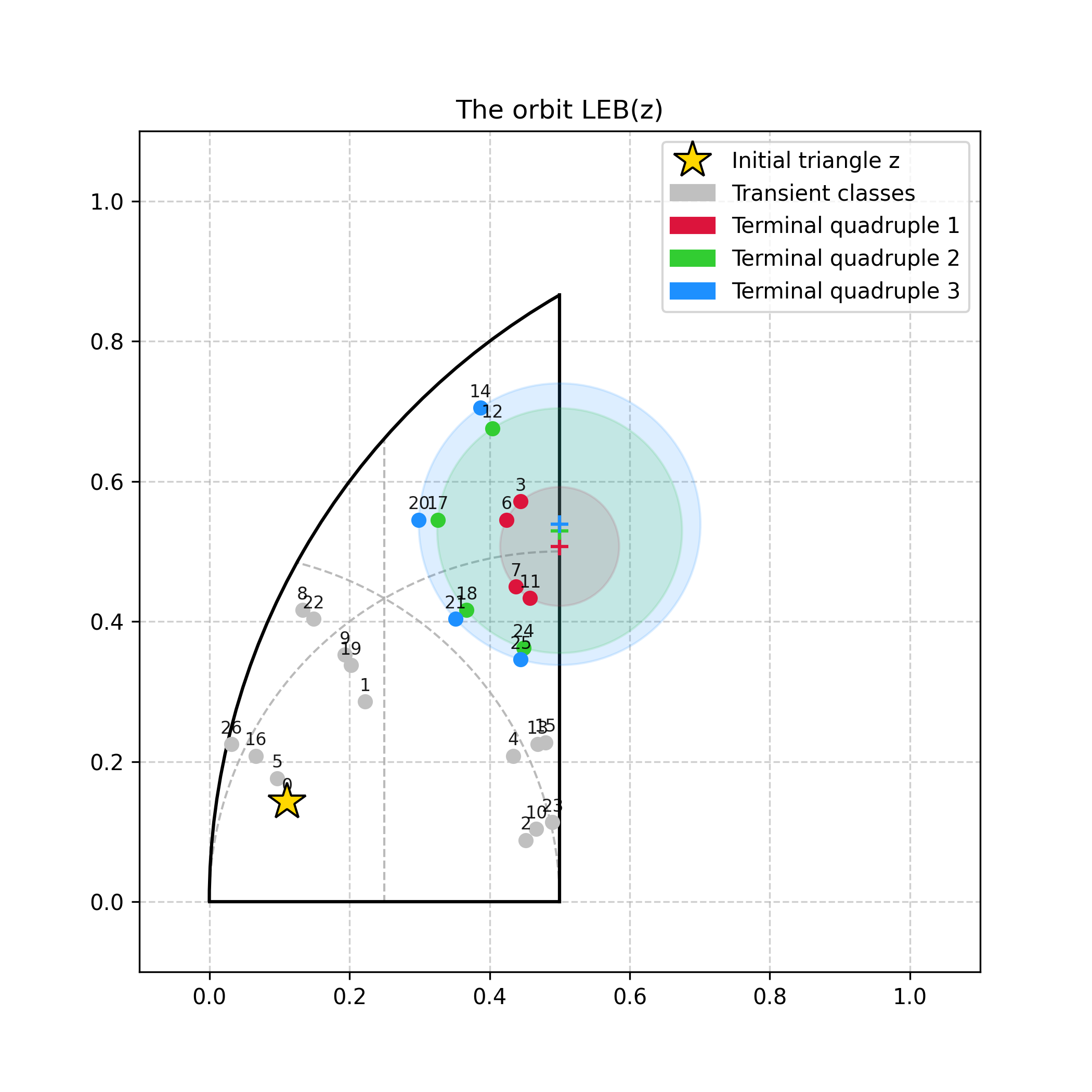}
    \caption{The points of $\leb(z)$ for $z=\frac{1}{9}+\frac{1}{7}i$. There are $3$ terminal quadruples, and the circles that they lie on are shaded.}\label{fig:final_points_plot_1:9+i:7}
\end{figure}

Among other geometric properties that one might notice in Figure~\ref{fig:final_points_plot_1:9+i:7} is the fact that the points labeled $3, 12, 14$ lie on a specific arc of a circle that is tangent to the real axis at $0$. We explain this fact in~\S\ref{subsec:leb_geo_dist}. 

\subsection{The flow towards $\cA$}\label{subsec:flow_to_A}
In Proposition~\ref{prop:flow_to-A} below, we show that the orbit of every point eventually reaches $\cA$. Since we did not find this fact explicitly in the literature, we provide a proof here. The proof relies on Lemma \ref{lem:L_closer_to_zeta}, showing that points get closer to $\cA$ under repeated applications of $L$. 

\begin{lemma}\label{lem:L_closer_to_zeta}
    Let $d:\HH\times\HH\to[0,\infty)$ denote the hyperbolic metric. For every $z\in \cD\setminus (\I\cup \IV)$ we have $d(L(z),\zeta)<d(z,\zeta)$. 
\end{lemma}

\begin{proof}
    Each piece $\tilde{L}$ of the piecewise function $L$ is an isometry of $\HH$, hence for every $z\in \cD$ we have $d(\tilde{L}(z),\zeta)=d(z,\tilde{L}^{-1}(\zeta))$. We proceed by case analysis for points $z$ in the different subregions of $\cD$, where in each case we compute the perpendicular bisector of $\zeta$ and $\tilde{L}^{-1}(\zeta)$ and use it to deduce that $z$ is closer to $\tilde{L}^{-1}(\zeta)$.
    \begin{itemize}
        \item 
        If $z\in\II\setminus \I$, we have $\tilde{L}(z)=\frac{-1}{2z-1}$. Then $\tilde{L}^{-1}(\zeta)=\frac{\zeta-1}{2\zeta} = \frac{i}{2}$ and the perpendicular bisector of $\zeta$ and $\tilde{L}^{-1}(\zeta)$ is the geodesic $\mRe(w)=\frac14$. 
        \item 
        If $z\in\III\setminus\I$, we have $\tilde{L}(z)=\frac{2\overline{z}}{2\overline{z}-1}$, and again $\tilde{L}^{-1}(\zeta)= \frac{\overline{\zeta}}{2\overline{\zeta}-2}=\frac{i}{2}$, and the perpendicular bisector is as in the previous case. 
        \item 
        If $z\in \V\setminus\I$,
        we have $\tilde{L}(z)=2z$. Then $\tilde{L}^{-1}(\zeta)=\frac{\zeta}{2}$ and the perpendicular bisector is the geodesic $\absolute{w}=\frac12$.
        \item 
        If $z\in \VI\setminus\IV$, 
        we have $\tilde{L}(z)=1-2\overline{z}$. Then $\tilde{L}^{-1}(\zeta)=\frac{1-\overline{\zeta}}{2} = \frac{\zeta}{2}$, and we proceed as in the previous case.
    \end{itemize}
    In all of these cases, clearly $z$ is closer to $\tilde{L}^{-1}(\zeta)$ than to $z$, which completes the proof. 
\end{proof}

\begin{proposition}\label{prop:flow_to-A}
    For every $z\in \cD$ there exists $n\in\N$ such that $L^n(z)\in\cA$. In particular, $\leb(z)\cap \cA\neq\emptyset$.
\end{proposition}

\begin{proof}
    Let $z\in\cD$ and assume for contradiction that $L^n(z)\notin\cA$ for every $n\in\N$. Then by the FSC theorem, the set $\set{L^n(z)}{n\in\N}$ is a finite subset of $\cD\setminus \cA$. So there exists some $n_0\in\N$ for which $L^{n_0}(z)$ minimizes the distance from $\zeta$. But Lemma~\ref{lem:L_closer_to_zeta} asserts that $L^{n_0+1}(z)$ is closer. A contradiction.
\end{proof}

Another consequence is that $\cA$ is exactly the locus of terminal quadruples. 

\begin{corollary}\label{cor:A_is_terminal}
    For every $z\in \cD$, if $\leb(z)$ is a terminal quadruple then $z\in\cA$.    
\end{corollary}

\begin{proof}
    Suppose that $\leb(z)$ is a terminal quadruple. By Proposition~\ref{prop:terminal_sets_are_on_a_circle}, the elements of $\leb(z)$ are at the same distance from $\zeta$. For contradiction, assume that $z\in\cD\setminus\cA$, and consider two cases. If $z\notin\IV$, by Lemma~\ref{lem:L_closer_to_zeta} we have $d(L(z),\zeta)<d(z,\zeta)$, a contradiction. If $z\in\IV$, then $L(z)\in\II$ and then $d(LL(z),\zeta)< d(z,\zeta)$, which is again a contradiction.   
\end{proof}

\subsection{Regions with $q(z)=1$}
We begin with some observations regarding compositions of $L$ and $R$ in certain regions. 
\begin{lemma}\label{lem:R_maps_I,II,III_to_IV}
For every $z\in \I \cup \II \cup \III$ we have $LR(z)=z$ and $RR(z)=\frac{\bar{z}}{2\bar{z}-1}$, which is a reflection with respect to the geodesic $\absolute{w-\frac12}=\frac12$.    
\end{lemma}

\begin{proof}
    Recall that for $z\in \I \cup \II \cup \III$ we have $R(z) = \frac{-1}{2z-2}$. One easily checks that $R$ maps the geodesics $\mRe(z)=\frac12$, $\absolute{z-1}=1$ and $\absolute{z-\frac12}=\frac12$ to the geodesics $\absolute{z-\frac12}=\frac12$, $\absolute{z}=\frac12$ and $\mRe(z)=\frac12$, respectively, and hence maps the region $\I \cup \II \cup \III$ to region $\IV$. Since $L(z)=\frac{2z-1}{2z}$ and $R(z)=\frac{2\bar{z}-1}{2\bar{z}-2}$ for $z\in\IV$, the assertion follows by a direct computation.  
\end{proof}

\begin{lemma}\label{lem:L=LRR_in_II_n_III}
For every $z\in \II \cup \III$ we have $L(z)=LRR(z)$.
\end{lemma}

\begin{proof}
   By Lemma \ref{lem:R_maps_I,II,III_to_IV}, for every $z\in \II$ we have $RR(z)\in \VI$ and thus 
   \[
   LRR(z) = 1-2\overline{\left(\frac{\bar{z}}{2\bar{z}-1}\right)} = 
   \frac{-1}{2z-1} = L(z).
   \]
   For $z\in\III$, notice that $L(z) = \frac{2\bar{z}}{2\bar{z}-1}= 2RR(z)$. Since $RR(z) \in \V$, where $L(z)=2z$, the proof is complete. 
\end{proof}

\begin{lemma}\label{lem:LL=LR_in_IV}
    For every $z\in \IV$ we have $LL(z) = LR(z)$.
\end{lemma}

\begin{proof}
    As before, we use subscripts to indicate which piece of the piecewise maps $L$ and $R$ is being applied.
    
    Since both $L$ and $R$ can map $z\in \IV$ to different regions, we distinguish between $3$ cases, determined by the geodesics $\absolute{z - \frac{1}{3}}= \frac{1}{3}$ and  $\absolute{z - \frac{2}{3}}= \frac{1}{3}$ in region $\IV$. 
    
    \subsection*{Case $1$: $\absolute{z - \frac{1}{3}} \ge \frac{1}{3}$.}
    One easily verifies that the inverse image of the geodesic $\absolute{z}=\frac{1}{2}$ under $R_{\IV\cup\V\cup\VI}$ is the geodesic $\absolute{z-\frac{1}{3}} = \frac{1}{3}$, which is also the inverse image of $\mRe(z) = \frac{1}{4}$ under $L_\IV$. Thus, for $z \in \IV$ that satisfies  $\absolute{z-\frac{1}{3}}\ge \frac{1}{3}$ we have $L(z)\in \I$ and $R(z)\in\IV$. Therefore, in this case, we have 
    \[
    L_\I(L_\IV(z)) = \frac{\overline{z}}{2\overline{z}-1} = L_\IV(\Rd(z)).
    \]

    \subsection*{Case $2$: $\absolute{z - \frac{1}{3}} < \frac{1}{3}$ and $\absolute{z - \frac{2}{3}} \ge \frac{1}{3}$.}
    Here, the inverse image of the geodesic $\mRe(z)=\frac{1}{4}$ under $R_{\IV\cup\V\cup\VI}$ is the geodesic $\absolute{z-\frac{2}{3}} = \frac{1}{3}$, which is also the inverse image of $\absolute{z} = \frac{1}{2}$ under $L_\IV$. Thus, for $z \in \IV$ that satisfies  $\absolute{z-\frac{2}{3}}\ge \frac{1}{3}$ we have $L(z)\in \II$ and $R(z)\in\VI$. Therefore
    \[
    L_\II(L_\IV(z)) = \frac{-z}{z-1} = L_\VI(\Rd(z)).
    \]

    \subsection*{Case $3$: $\absolute{z - \frac{1}{3}} < \frac{1}{3}$ and $\absolute{z - \frac{2}{3}} < \frac{1}{3}$. }
    Finally, by the reasoning in case $2$, for $z \in \IV$ that satisfies  $\absolute{z-\frac{1}{3}}< \frac{1}{3}$ we have $L(z)\in \III$ and $R(z)\in\V$. Hence 
    \[
    L_\III(L_\IV(z)) = \frac{2\overline{z}-1}{\overline{z}-1} = L_\V(\Rd(z)).
    \]
\end{proof}

With these identities at hand we are ready to prove Theorem~\ref{thm:q(z)=1}.

\begin{proof}[Proof of Theorem~\ref{thm:q(z)=1}]
    We proceed by case analysis. Notice that for every triangle $z$, the orbit of $z$ can be written as follows:  
    \begin{equation}\label{eq:orbit_partition}
        \begin{split}
            \leb(z) &= \{z\} \cup \leb(L(z)) \cup \leb(R(z)) \\
            &= \{z\} \cup \leb(L(z)) \cup\leb(LR(z)) \cup \leb(RR(z)) \cup \{R(z)\}.
        \end{split}
    \end{equation}
    
    \subsection*{Assume $z\in \I$} Then $z$ belongs to the terminal quadruple $$\leb(z)=\{z,L(z),R(z), RR(z)\}$$ and clearly $q(z)=1$ (see e.g.~\cite[Lemma 4.4, Figure 8]{PerdP14}).

    \subsection*{Assume $z\in \II$} Then $RR(z)\in \VI$ by Lemma \ref{lem:R_maps_I,II,III_to_IV}. Since $R$ maps the region $\IV\cup \V \cup \VI$ to itself, by an inversion with respect to $\absolute{z-1}=\frac{\sqrt2}{2}$, we have $RRR(z)=R(z)$. By Lemma \ref{lem:R_maps_I,II,III_to_IV} we have $LR(z)=z$ and by Lemma \ref{lem:L=LRR_in_II_n_III} we have $ LRR(z)=L(z)$. Applying these identities to \eqref{eq:orbit_partition} yields
    \begin{equation}\label{eq:orbit_partition_II_n_III}
        \begin{split}
            \leb(z) &= \{z\} \cup \leb(L(z)) \cup \leb(LRR(z)) \cup \{R(z)\} \cup \{RR(z)\}  \\
            &= \{z\}\cup \{R(z)\} \cup \{RR(z)\} \cup \leb(L(z)).
        \end{split}
    \end{equation}
    By Corollary~\ref{cor:A_is_terminal}, the triangles $z, R(z), RR(z)$ do not belong to a terminal quadruple and $L(z)\in \I$, it follows that $q(z)=1$. 
    
    \subsection*{Assume $z\in \III$} Then using the same identities as in the previous case, \eqref{eq:orbit_partition_II_n_III} holds.
    As before, the triangles $z, R(z), RR(z)$ do not belong to a terminal quadruple. Thus, terminal quadruples can only arise from $\leb(L(z))$. Notice that $L(z)\in \I \cup \II \cup \III$. If $L(z)\in \I \cup \II$ then we are done, by the previous cases. Otherwise, $L(z)\in \III$, and we may repeat the above argument. Namely, the singletons $\{L(z)\}, \{RL(z)\}$ and $\{RRL(z)\}$, which do not belong to a terminal quadruple, are added to $\leb(z)$, and terminal quadruples can only arise from $\leb(LL(z))$. By Proposition \ref{prop:flow_to-A}, there exists some $k\in\N$ for which $L^k(z)\in \cA$, and we are done. 
        
    \subsection*{Finally, assume $z\in \IV$.} We have $RR(z)=z$, and so~\eqref{eq:orbit_partition} in this case reads:
    \[
        \leb(z) = \{z\} \cup \leb(L(z)) \cup\leb(LR(z)) \cup \{R(z)\}.
    \]
    Since $L(z) \in\I \cup \II \cup \III$ for $z\in\IV$, by the previous cases, $\leb(L(z))$ contains a single terminal quadruple. By Lemma \ref{lem:LL=LR_in_IV}, we have $\leb(LR(z)) = \leb(LL(z))\subset \leb(L(z))$ and the proof is complete.   
\end{proof}

\subsection{Arbitrarily large $q(z)$}
In this subsection, we prove Theorem \ref{thm:arbitrarily_many_quadruples}, showing that the number of terminal quadruples that arise from an initial triangle can be arbitrarily large. For the proof, we focus on a subregion of $\V$ defined by 
\[
\tilde \V = \set{z\in \V}{ \absolute{z-\frac13}\le \frac13}.
\]
The geodesic $\absolute{z-\frac13} = \frac13$ is the inverse image of the geodesic $\absolute{z}=\frac12$ under the function $R$. The relevance of this region will be made clear in the proof of Lemma \ref{lem:geometry_of_h}. 

For a triangle $z\in \tilde \V$ we denote by 
\[
h(z)=LR\left(z\right)
\]
and begin with some straightforward properties of the function $h$. 

\begin{lemma}\label{lem:geometry_of_h}
For $z\in \tilde\V$ we have 
\begin{enumerate}[(a)]
    \item\label{lemitem_a:geometry_of_h}
$h(z)=\frac{-z}{z-1}$. 
    \item\label{lemitem_b:geometry_of_h} 
The function $h$ is a parabolic M\"obius transformation, with a source and sink at the fixed point $w=0$. Moreover, for every $k\in\N$, the point $h^{k+1}(z)$ is obtained from $h^k(z)$ by a counterclockwise rotation along the circle $C_z$ through $z$, that is tangent to the real axis at $0$. 
    \item\label{lemitem_c:geometry_of_h}
If $\mRe(z)=\mIm(z)$ then the tangent line to $C_z$ at $z$ is vertical.      
\end{enumerate}
\end{lemma}

\begin{proof}
    First note that for every $z\in \tilde \V$ we have $R(z) \in \VI$, thus property (\ref{lemitem_a:geometry_of_h}) follows directly from the formulas of $R$ in region $\V$ and $L$ in region $\VI$. In view of (\ref{lemitem_a:geometry_of_h}), property (\ref{lemitem_b:geometry_of_h}) follows by standard analysis, see e.g. \cite[p.88]{Ahlfors}. For property (\ref{lemitem_c:geometry_of_h}), let $o_z$ denote the center of $C_z$. Since $C_z$ is tangent to the real axis at $0$, the point $o_z$ is on the imaginary axis. Since $C_z$ passes through $z$, which satisfies $\mRe(z)=\mIm(z)$, one easily shows that the Euclidean distance from $o_z$ to $z$ is $\mIm(z)$ and from $o_z$ to $0$ is $\mRe(z)$, which implies the assertion.
\end{proof}

\begin{lemma}\label{lem:h_commuting_diagram}
For every $z$ for which $z,h(z), L(z)\in \tilde\V$ we have $L\circ h^2(z) = h\circ L(z)$. 
\end{lemma}

\begin{proof}
    By the previous lemma, and since $z,h(z), L(z)\in \tilde\V$, one easily verifies that $h^2(z)\in\V$ and that
    \[
    L\circ h^2(z) = \frac{-2z}{2z-1} = h\circ L(z),
    \]
    implying the assertion.
\end{proof}

For $k\in \N$ we consider the region $B_k = \frac{1}{2^k} \I$, and set 
\[
B=\bigcup_{k=1}^{\infty} B_k=\bigcup_{k=1}^{\infty} \frac{1}{2^k}\I.
\]
Notice that $B$ is comprised of shrinking copies of region $\I$ that approach the origin (see Figure~\ref{fig:B}), in particular, $B\subseteq \V$. Furthermore, since $\I\subseteq \set{z}{\absolute{z-1}\le 1}$, for each $k\ge 1$ we have $B_k\subseteq \set{z}{\absolute{z-\frac{1}{2^k}}\le \frac{1}{2^k}}$, and in particular, $B_k\subseteq \tilde \V$ for every $k\ge 2$. 

\begin{lemma}\label{lem:2^m_points_in_B}
    For $\zeta = \frac12+\frac12 i$, and every $m\in \N$ the point $z_m=\frac{\zeta}{2^m}\in B_m$ satisfies 
    \[
    \set{ h^j(z_m)}{j\in \{0,\ldots,2^{m-1}\} } \subseteq B_m.
    \]
\end{lemma}

\begin{proof}
    We argue by induction on $m$. For $m=1$, note that $z_1 = \frac14+\frac14 i \in \tilde \V$, then $h(z_1) = \frac15+\frac25 i \in B_1$ (since $|h(z_1)-\frac12|=\frac12$ and the geodesic $\absolute{w-\frac12}=\frac12$ bounds $B_1$ from above, see Figure \ref{fig:B}), as required. 
    For $z_m\in B_m$ suppose that  $\set{h^j(z_m)}{j\in \{0,\ldots,2^{m-1}\} } \subseteq B_m$. For the induction step, since $B_k\subseteq \tilde\V$ for every $k\ge 2$, by Lemma \ref{lem:h_commuting_diagram}, we have 
    \begin{equation}\label{eq:B_m}
    h^{2^{m}}(z_{m+1})=(L^{-1}\circ h\circ L)^{2^{m-1}}(z_{m+1}) = L^{-1}\circ h^{2^{m-1}}\circ L(z_{m+1}) = \frac12 h^{2^{m-1}}(z_m) \in B_{m+1}, 
   \end{equation}
    where the relation on the right-hand side follows from $h^{2^{m-1}}(z_m)\in B_m$, which is obtained from the induction hypothesis. In particular, $h^{2^{m}}(z_{m+1})\in B_{m+1}$. Moreover, iterating \eqref{eq:B_m} yields 
    \[
    h^{2^{m}}(z_{m+1}) = \frac{1}{2^{m}} h(z_1).
    \]
    Therefore, $\absolute{h^{2^{m}}(z_{m+1}) - \frac{1}{2^{m+1}}} = \frac{1}{2^m}\absolute{h(z_1)-\frac12} = \frac{1}{2^{m+1}}$, which shows that $h^{2^{m}}(z_{m+1})$ is on the top boundary of $B_{m+1}$, see Figure \ref{fig:B}.
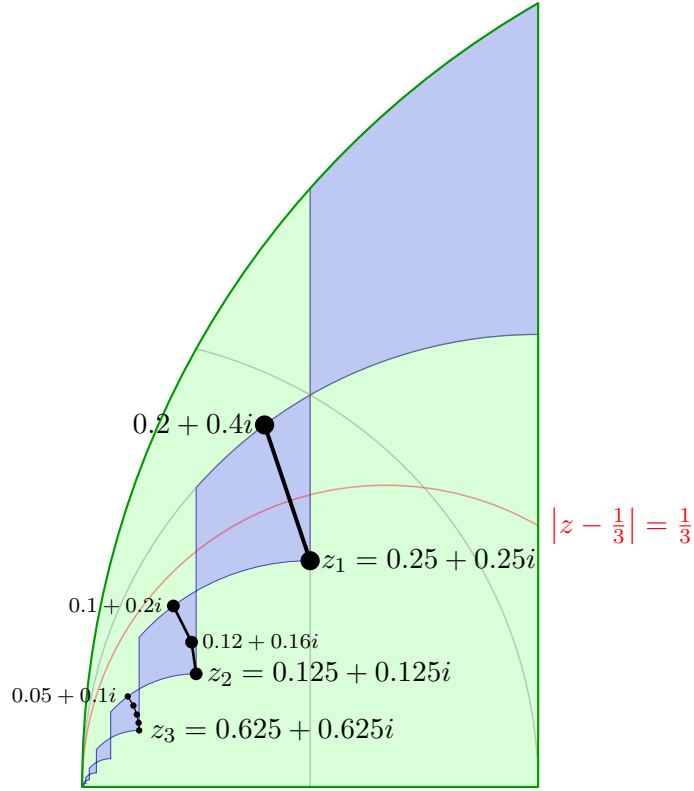
\begin{figure}[h!]
\centering
\begin{tikzpicture}[scale=12] 

    \coordinate (A) at (0,0);
    \coordinate (B) at (0.5,0);
    \coordinate (C) at (0.5, {sqrt(3)/2});

    \fill[green!15] (A) -- (B) -- (C) arc (120:180:1) -- cycle;

    \foreach \n in {0,1,2,3,4,5,6,7,8} {
        \pgfmathsetmacro{\s}{1/(2^\n)}
        \pgfmathsetmacro{\xVone}{\s * 0.25}
        \pgfmathsetmacro{\yVone}{\s * sqrt(7)/4}
        \pgfmathsetmacro{\xVtwo}{\s * 0.5}
        \pgfmathsetmacro{\yVtwo}{\s * sqrt(3)/2}
        \pgfmathsetmacro{\xVthree}{\s * 0.5}
        \pgfmathsetmacro{\yVthree}{\s * 0.5}
        \pgfmathsetmacro{\xVfour}{\s * 0.25}
        \pgfmathsetmacro{\yVfour}{\s * sqrt(3)/4}
        \coordinate (sV1) at (\xVone, \yVone);
        \coordinate (sV2) at (\xVtwo, \yVtwo);
        \coordinate (sV3) at (\xVthree, \yVthree);
        \coordinate (sV4) at (\xVfour, \yVfour);
        \pgfmathsetmacro{\angleArcOneStart}{atan2(\yVtwo, \xVtwo - \s * 1)}
        \pgfmathsetmacro{\angleArcOneEnd}{atan2(\yVone, \xVone - \s * 1)}
        \pgfmathsetmacro{\angleArcTwoStart}{atan2(\yVfour, \xVfour - \s * 0.5)}
        \pgfmathsetmacro{\angleArcTwoEnd}{90}
        \pgfmathsetmacro{\radiusOne}{\s}
        \pgfmathsetmacro{\radiusTwo}{\s/2}
        \path[fill=blue!30, opacity=0.7, draw=blue!60!black, line width=0.4pt]
            (sV3) -- (sV2)
            arc (\angleArcOneStart:\angleArcOneEnd:{\radiusOne})
            -- (sV4)
            arc (\angleArcTwoStart:\angleArcTwoEnd:{\radiusTwo})
            -- cycle;
    }

    \draw[black!50, line width=0.6pt, opacity=0.4]
        (0.25, 0) -- (0.25, {sqrt(7)/4}) 
        (0.5,0) arc (0:{acos(0.25)}:0.5) 
        (0,0) arc (180:90:0.5); 

    \draw[red, line width=0.6pt, opacity=0.4] 
        (0,0) arc (180:60:{1/3}) 
        node[pos=1, right, font=\small, text=red, text opacity=1] {$\absolute{z-\frac13}=\frac13$};
     
    \draw[green!60!black, line width=0.8pt, line join=round]
      (A) -- (B) -- (C) arc (120:180:1) -- cycle;

    \coordinate (Z0) at (0.25, 0.25);
    \coordinate (Z1) at (0.2, 0.4);
    
    \draw[black, line width=1.5pt, -] (Z0) -- (Z1);

    \fill[black] (Z0) circle (0.3pt) node[anchor=west, font=\small] {$z_1=0.25+0.25i$};
    \fill[black] (Z1) circle (0.3pt) node[anchor=east, font=\small] {$0.2+0.4i$};
    
    \coordinate (Z0) at (0.125, 0.125);
    \coordinate (Z1) at (0.12, 0.16);
    \coordinate (Z2) at (0.1, 0.2);
    
    \draw[black, line width=1pt, -] (Z0) -- (Z1) -- (Z2);

    \fill[black] (Z0) circle (0.2pt) node[anchor=west, font=\small] {$z_2=0.125+0.125i$};
    \fill[black] (Z1) circle (0.2pt) node[anchor=west, font=\tiny] {$0.12+0.16i$};
    \fill[black] (Z2) circle (0.2pt) node[anchor=east, font=\tiny] {$0.1+0.2i$};

    \coordinate (Z0) at (0.0625, 0.0625);
    \coordinate (Z1) at (0.0619, 0.0708);
    \coordinate (Z2) at (0.06, 0.08);
    \coordinate (Z3) at (0.0562, 0.0899);
    \coordinate (Z4) at (0.05, 0.1);
    
    \draw[black, line width=0.5pt, -] (Z0) -- (Z1) -- (Z2) -- (Z3) -- (Z4);

    \fill[black] (Z0) circle (0.1pt) node[anchor=west, font=\small] {$z_3=0.625+0.625i$};
    \fill[black] (Z1) circle (0.1pt);
    \fill[black] (Z2) circle (0.1pt);
    \fill[black] (Z3) circle (0.1pt);
    \fill[black] (Z4) circle (0.1pt) node[anchor=east, font=\tiny] {$0.05+0.1i$};
\end{tikzpicture}
\centering
\caption{The region $B$ is shaded in blue. The black points are the parts of the orbits of $z_1 = \frac{1}{4}+\frac{1}{4}i$, $z_2 = \frac{1}{8}+\frac{1}{8}i$ and $z_3 = \frac{1}{16}+\frac{1}{16}i$ that belong to $B$.}
\label{fig:B}
\end{figure}
    Finally, denote by $C_{m+1}$ the circle through $0$ and $z_{m+1}$ that contains the $h$-orbit of $z_{m+1}$. By \emph{\eqref{lemitem_c:geometry_of_h}} of Lemma \ref{lem:geometry_of_h}, the tangent to $C_{m+1}$ at $z_{m+1}$ is vertical. Since $h^{2^{m}}(z_{m+1})$ is on the top boundary of $B_{m+1}$, the region $B_{m+1}$ contains the entire arc of $C_{m+1}$ that connects $z_{m+1}$ and $h^{2^m}(z_{m+1})$. Since, by \emph{\eqref{lemitem_b:geometry_of_h}} of Lemma \ref{lem:geometry_of_h}, $h$ acts by moving a point counterclockwise on the circle $C_{m+1}$, we have $\set{h^j(z_{m+1})}{ j\in \{0,\ldots,2^{m}\}} \subseteq B_{m+1}$, which completes the proof.
\end{proof}

We are now ready to deduce Theorem \ref{thm:arbitrarily_many_quadruples}.
\begin{proof}[Proof of Theorem \ref{thm:arbitrarily_many_quadruples}]
   We begin with finding such points $z$ in region $\V$. Given $n\in\N$ we choose $m\in\N$ that satisfies $2^{m}\ge n$, and set $z= z_{m+2}$ from Lemma \ref{lem:2^m_points_in_B}. Then 
   \[
   \set{h^j(z)}{j\in \{0,\ldots,2^{m+1}\}}\subseteq B_{m+2}\subseteq B.
   \] 
   Note that by \emph{\eqref{lemitem_b:geometry_of_h}} of Lemma \ref{lem:geometry_of_h}, the set $\set{h^j(z)}{ j\in \{0,\ldots,2^{m+1}\}}$ contains $2^{m+1}+1$ distinct points. Also, notice that for every $z\in B$ there exists a minimal $k\in \N$ for which $L^k(z)\in \I$, and in particular, for every $j\in \{0,\ldots,2^{m+1}\}$ we have $L^{m+2}(h^j(z))\in \I$. These are $2^{m+1}+1$ distinct triangles in region $\I$ that are in $\leb(z)$. Each of these triangles in $\I$ has an orbit of length at most $4$, which contains at most $2$ elements from region $\I$. It follows that $q(z)\ge \frac{2^{m+1}+1}{2}>2^m\ge n$, as required.
   
   To find such $z$ in region $\VI$, since $R^2=\mathrm{Id}$ in regions $\V$ and $\VI$, it suffices to show that $R(z_m) = R\left(\frac{\zeta}{2^m}\right)\in \VI$, for $m$ large enough. 
   Since
   \begin{equation*}
       \begin{split}
            & R\left(\frac{\zeta}{2^m}\right) = 
            1+\frac12\frac{1}{\frac{\overline\zeta}{2^m}-1} = 
            \left(\frac12 - \frac{2^{m+1}-2}{2^{2m+2}-2^{m+2}}\right) + i\left(\frac{1}{2^{m+1}-2}\right),
   \end{split}
   \end{equation*}
   we clearly have $\absolute{R\left(\frac{\zeta}{2^m}\right) - \frac12} < \frac12$ and $\mRe\left(R\left(\frac{\zeta}{2^m}\right)\right)>\frac14$, for every $m\ge 2$, and hence $R\left(\frac{\zeta}{2^m}\right)\in\VI$ and the proof is complete. 
\end{proof}

\section{The bisection graph and its spectral properties}\label{section:graph}
In this section, we prove Theorem~\ref{thm:occupying_area}. The proof relies on known results for nonnegative matrices. We state them here for the completeness of the text. 
We use the notation $\mnull(A)$ for the null space of a matrix $A$, and $[n] = \left\{1,2,\ldots,n\right\}$. Recall that the {\bf spectral radius} $\rho(A)$ of a matrix $A$ is the maximal absolute value of an eigenvalue of $A$. For a graph $G$, we use $\rho(G)$ for the spectral radius of its adjacency matrix. By the {\bf generalized eigenspace of $A$ that belongs to an eigenvalue $\lambda$} we refer to the space $\mnull(A-\lambda I)^{d(\lambda)}$, where $d(\lambda)$ is the the smallest natural number $k$ that satisfies $\mnull(A-\lambda I)^{k} = \mnull(A-\lambda I)^{k+1}$. The non-zero vectors of the generalized eigenspace are called {\bf generalized eigenvectors}.   

\begin{theorem}\label{thm:bounds_on_rho(A)}\cite[Theorem 8.1.22]{HornJ13}
    Let $A=(a_{ij})$ be a nonnegative $n\times n$ matrix, then 
    \begin{equation*}
        \begin{split}
            \min_{1\le i\le n}\sum_{j=1}^na_{ij} \le  \rho(A) \le  
    \max_{1\le i\le n}\sum_{j=1}^na_{ij} 
    ~~~\text{ and }~~~
    \min_{1\le j\le n}\sum_{i=1}^na_{ij} \le  \rho(A) \le  \max_{1\le j\le n}\sum_{i=1}^na_{ij}.
    \end{split}
    \end{equation*}
\end{theorem}

The following theorem is a particular case of \cite[Theorem 3.20, p. 43]{BermP94}.
\begin{theorem}\label{thm:basis_of_eigenvectors}
Suppose that $G=\left([n],E\right)$ is a directed graph and $A_G$ is the adjacency matrix of $G$. Let $\alpha_1,\ldots,\alpha_s$ be all the strongly connected components of $G$ with $\rho(\alpha_j)=\rho(A_G)$, for every $j$. Then the dimension of the generalized eigenspace of $\rho(A_G)$ is at most $s$. 
\end{theorem}
 
As a consequence of the Perron--Frobenius theorem (see \cite[Theorem 8.4.4]{HornJ13}) we have:

\begin{lemma}\label{lem:non-final_irreducible_blocks}
    Let $A\in M_n(\R)$ be a nonnegative matrix with column sums equal to $\lambda>0$. Let $B\in M_k(\R)$, with $2\le k<n$, be an irreducible principal submatrix of $A$ with at least one column sum which is strictly smaller than $\lambda$. Then $\rho(B)<\lambda = \rho(A)$.  
\end{lemma}

\begin{proof}
    Since $B$ is irreducible, by the Perron--Frobenius theorem, there is a positive eigenvector $v=(c_1,\ldots,c_k)^t$ (i.e., $c_j>0$ for all $j$) satisfying $Bv=\rho(B)v$. The assumptions on the column sums of $A$ and $B$ imply that the vector $B^t \mathds{1}$ satisfies the inequality $B^t \mathds{1} \le \lambda \mathds{1}$ entry-wise, with strict inequality in at least one component, and since $v$ is positive, we have:
    \[
    \langle B^t \mathds{1}, v \rangle < \langle \lambda \mathds{1}, v \rangle.
    \]
    
    Thus we get:
    \[
    \rho(B) \langle \mathds{1}, v \rangle = \langle \mathds{1}, \rho(B)v \rangle = \langle \mathds{1}, Bv \rangle = \langle B^t  \mathds{1}, v \rangle < \langle \lambda \mathds{1}, v \rangle = \lambda \langle \mathds{1}, v \rangle,
    \]
    dividing by the positive scalar $\langle \mathds{1}, v \rangle$, we conclude that $\rho(B) < \lambda=\rho(A)$.
\end{proof}

We are now ready to introduce the bisection graph $\gleb(z)$. We start with the generic case.

\begin{definition}[generic $\gleb(z)$]
    For a generic initial triangle $z$, the $\leb$ process can be represented by a directed graph $\gleb(z)=(V,E)$ on the vertex set $V=\leb(z)$, with 
    \[
    E=\set{\left(w,L(w)\right)}{w\in\leb(z)} \uplus \set{\left(w,R(w)\right)}{w\in\leb(z)}.
    \]
    We label each edge in $\gleb(z)$ naturally by $L$ or $R$. 
\end{definition}

\begin{figure}[h!]
    \centering
    \includegraphics[scale=0.081]{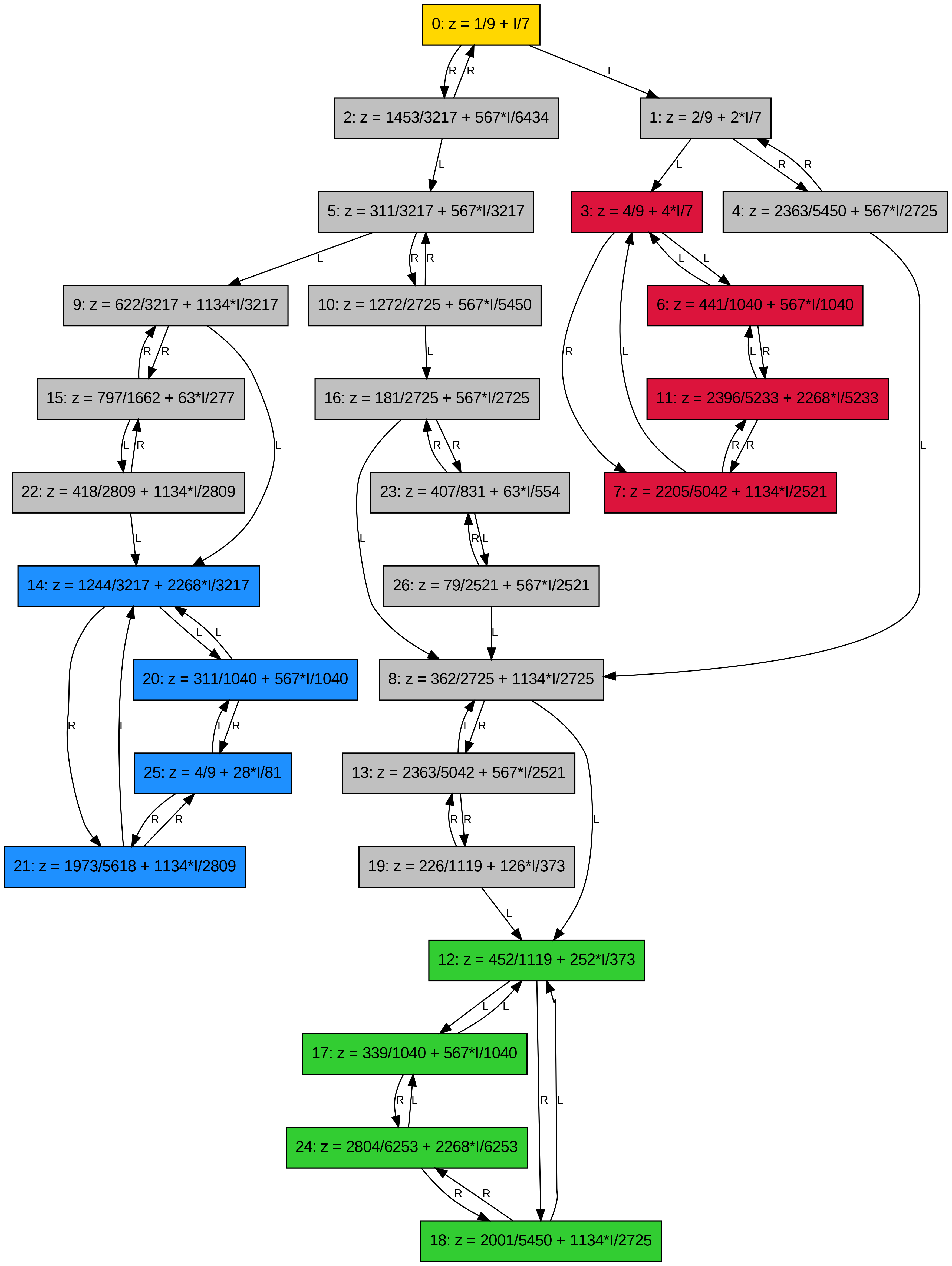}
    \caption{The graph $\gleb(z)$ for $z=\frac{1}{9}+\frac{1}{7}i$. The colors and numbering correspond to what is shown in Figure~\ref{fig:final_points_plot_1:9+i:7}.}
    \label{fig:bisection_graph_1:9+i:7}
\end{figure}

For non-generic $z$, the above graph has loops or double edges. From Proposition~\ref{prop:non-generic}, we know that the set $T=\set{w\in\leb(z)}{\absolute{w}=\frac{\sqrt{2}}{2}\text{ or }\absolute{w-1}=\frac{\sqrt{2}}{2}}$ is the subset of vertices that have loops, and the set $S=\set{w'\in\leb(z)}{\mRe(w')=\frac{1}{2}}$ is the subset of vertices that have (outgoing) double edges. Although the set of non-generic points has measure zero, it is desirable for $\gleb(z)$ to be a simple directed graph even in this case. To this end, we introduce two vertices to represent a ``problematic'' point, having the same outgoing edges, and treat the incoming edges as in the following definition. This procedure formally makes $\leb(z)$ larger.

\begin{definition}[non-generic $\gleb(z)$]
    Let $z \in \cD$ be a non-generic triangle.
    \begin{enumerate}
        \item {\bf Resolving a single loop ($T \setminus A$)}: A point $w$ on this curve has a loop $R(w)=w$. We replace $w$ with two vertices $w_1, w_2$, and the loop is replaced by the $2$-cycle: edges $(w_1, w_2)$ labeled $R$ and $(w_2, w_1)$ labeled $R$. Incoming edges to $w$ are directed to $w_1$. See Figure~\ref{fig:non-generic_1}.
        
        \item {\bf Resolving a pair of loops ($T \cap A$)}: Points in $T \cap A$ come in pairs: $w, w'$, where each possesses a loop $L(w)=w, R(w')=w'$, and $L(w')=w, R(w)=w'$. We apply the splitting procedure (as in the previous item) to both vertices. We then connect the resulting pairs by duplicating the edges between $w$ and $w'$: for $i=1,2$, we add edges $(w_i, w'_i)$ labeled $R$ and $(w'_i, w_i)$ labeled $L$. Incoming edges to $w$ (resp. $w'$) are directed to $w_1$ (resp. $w'_1$). See Figure~\ref{fig:non-generic_2}.
        
        \item {\bf Resolving double edges ($S$)}: For a point $w'\in S$ where $L(w')=R(w')=w$, we split $w$ into two vertices $w_L$ and $w_R$. All incoming edges labeled $L$ are directed to $w_L$, and all incoming edges labeled $R$ are directed to $w_R$. See Figure~\ref{fig:non-generic_3}.
    \end{enumerate}
\end{definition}

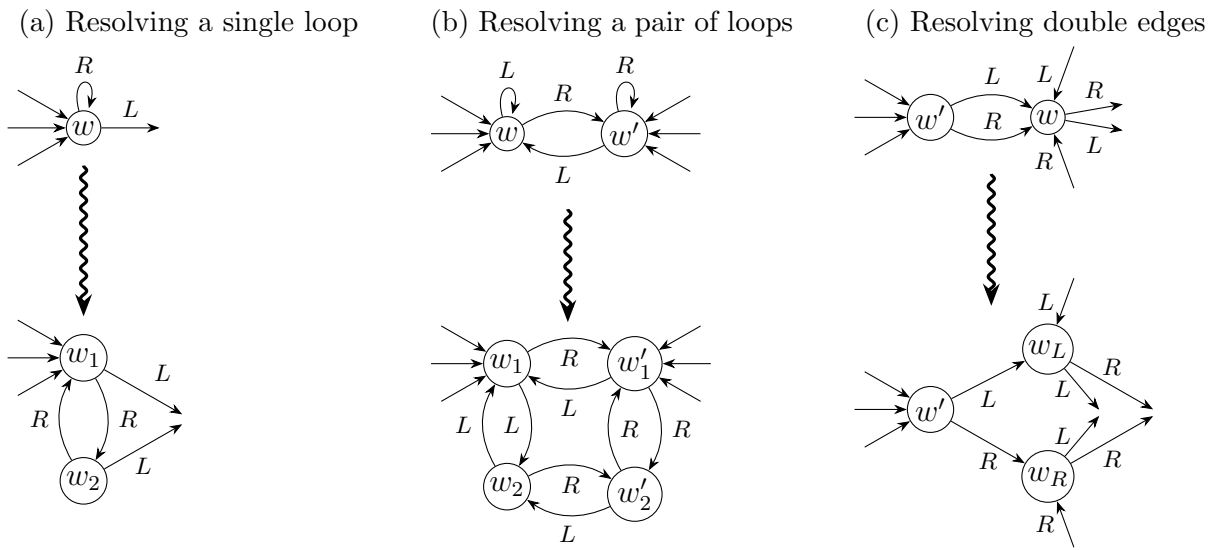
\begin{figure}[h!]
    \centering
    \begin{minipage}[t]{.3\textwidth}
    \subcaption{Resolving a single loop}
    \begin{tikzpicture}[
        baseline=(current bounding box.center),
        node distance = 1cm,
        >=Stealth,
        auto,
        state/.append style={inner sep=1pt, minimum size=0pt} 
        ]
    
        \node[state] (w) {$w$};
    
        \path[->, font=\scriptsize] 
            (w) edge[loop above] node {$R$}
            (w) edge node {$L$} ++(1,0);
        \path[<-, font=\scriptsize]
            (w) edge ++(150:1cm)
            (w) edge ++(180:1cm)
            (w) edge ++(210:1cm);
    
        \draw[->, very thick, decorate, decoration={snake, amplitude=.4mm, segment length=2mm, post length=1mm}] 
            (0,-0.5) -- (0,-2.5);
    
        \node[state, below=2.5cm of w] (w1) {$w_1$};
        \node[state, below=of w1] (w2) {$w_2$};
    
        \path[->, font=\scriptsize] 
            (w1) edge[bend left] node {$R$} (w2)
            (w2) edge[bend left] node {$R$} (w1)
            (w1) edge node {$L$} ++(-30:1.5cm)
            (w2) edge node[below] {$L$} ++(30:1.5cm);
        \path[<-, font=\scriptsize]
            (w1) edge ++(150:1cm)
            (w1) edge ++(180:1cm)
            (w1) edge ++(210:1cm);
    \end{tikzpicture}
    \label{fig:non-generic_1}
    \end{minipage}
    \hfill
    \begin{minipage}[t]{.3\textwidth}
    \subcaption{Resolving a pair of loops}
    \begin{tikzpicture}[
        baseline=(current bounding box.center),
        node distance = 1cm,
        >=Stealth,
        auto,
        state/.append style={inner sep=1pt, minimum size=0pt} 
        ]
    
        \node[state] (w) at (-0.8,0) {$w$};
        \node[state, right=of w] (w') {$w'$};

        \path[->, font=\scriptsize] 
            (w) edge[loop above] node {$L$} (w)
            (w) edge[bend left] node {$R$} (w')
            (w') edge[bend left] node {$L$} (w)
            (w') edge[loop above] node {$R$} (w');
        \path[<-, font=\scriptsize]
            (w) edge ++(150:1cm)
            (w) edge ++(180:1cm)
            (w) edge ++(210:1cm)
            (w') edge ++(150:-1cm)
            (w') edge ++(180:-1cm)
            (w') edge ++(210:-1cm);
        
        \draw[->, very thick, decorate, decoration={snake, amplitude=.4mm, segment length=2mm, post length=1mm}] 
            (0,-1) -- (0,-2.5);
    
        \node[state, below=2.5cm of w] (w1) {$w_1$};
        \node[state, below=of w1] (w2) {$w_2$};
        \node[state, right=of w1] (w'1) {$w'_1$};
        \node[state, below=of w'1] (w'2) {$w'_2$};
    
        \path[->, font=\scriptsize] 
            (w1) edge[bend left] node[left] {$L$} (w2)
            (w2) edge[bend left] node[left] {$L$} (w1)
            (w'1) edge[bend left] node[right] {$R$} (w'2)
            (w'2) edge[bend left] node[right] {$R$} (w'1)
            (w1) edge[bend left] node[below] {$R$} (w'1)
            (w2) edge[bend left] node[below] {$R$} (w'2)
            (w'1) edge[bend left] node[below] {$L$} (w1)
            (w'2) edge[bend left] node[below] {$L$} (w2);
        \path[<-, font=\scriptsize]
            (w1) edge ++(150:1cm)
            (w1) edge ++(180:1cm)
            (w1) edge ++(210:1cm)
            (w'1) edge ++(150:-1cm)
            (w'1) edge ++(180:-1cm)
            (w'1) edge ++(210:-1cm);
    \end{tikzpicture}
    \label{fig:non-generic_2}
    \end{minipage}
    \hfill
    \begin{minipage}[t]{.3\textwidth}
    \subcaption{Resolving double edges}
    \begin{tikzpicture}[
        baseline=(current bounding box.center),
        node distance = 1cm,
        >=Stealth,
        auto,
        state/.append style={inner sep=1pt, minimum size=0pt} 
        ]

        \node[state] (w') at (-0.8,0) {$w'$};
        \node[state, right=of w'] (w) {$w$};
            
        \path[->, font=\scriptsize] 
            (w') edge[bend left] node {$L$} (w)
            (w') edge[bend right] node {$R$} (w)
            (w) edge node[below] {$L$} ++(170:-1cm)
            (w) edge node[above] {$R$} ++(190:-1cm);
        \path[<-, font=\scriptsize]
            (w') edge ++(150:1cm)
            (w') edge ++(180:1cm)
            (w') edge ++(210:1cm)
            (w) edge node[left] {$L$} ++(70:1cm)
            (w) edge node[left] {$R$} ++(-70:1cm);
    
        \draw[->, very thick, decorate, decoration={snake, amplitude=.4mm, segment length=2mm, post length=1mm}] 
            (0,-0.75) -- (0,-2.5);

        \node[state, below=3.25cm of w'] (w'') {$w'$};
        \node[state, below=2.5cm of w] (wL) {$w_L$};
        \node[state, below=of wL] (wR) {$w_R$};
    
        \path[->, font=\scriptsize] 
            (w'') edge node[below] {$L$} (wL)
            (w'') edge node[below] {$R$} (wR)
            (wL) edge node[left] {$L$} ++(-50:1.05cm)
            (wR) edge node[left] {$L$} ++(50:1.05cm)
            (wL) edge node[above] {$R$} ++(-30:1.6cm)
            (wR) edge node[below] {$R$} ++(30:1.6cm);
        \path[<-, font=\scriptsize]
            (w'') edge ++(150:1cm)
            (w'') edge ++(180:1cm)
            (w'') edge ++(210:1cm)
            (wL) edge node[left] {$L$} ++(70:1cm)
            (wR) edge node[left] {$R$} ++(-70:1cm);
    \end{tikzpicture}
    \label{fig:non-generic_3}
    \end{minipage}
    \caption{Adjustments of $\gleb(z)$ for vertices with loops or double edges.}
\end{figure}

Notice that with the above definition, even for non-generic points, a terminal quadruple corresponds to a directed $K_{2,2}$ subgraph of $\gleb(z)$ (this fact, in the generic case, already appeared in~\cite[Lemma 4]{RoseS75}).

For a triangle $z$, we enumerate the similarity classes in $\leb(z)$ by $[n]$, setting the similarity class of $z$ first, and consider the \textbf{adjacency matrix} $A\in M_{n}\left(\R\right)$ of $\gleb(z)$ defined by 
\[
A_{ij}=\begin{cases}
    1, &(j,i)\in E\\
    0, &\text{otherwise}
\end{cases}
\]
Note that our definition of the adjacency matrix is the transpose of the common definition.

We use the vector $v_S\in \Z_+^n$ to encode a collection $S$ of triangles in $\leb(z)$, where the $i$th coordinate of $v_S$ represents the number of triangles of type $i$ in $S$. It is clear that multiplying $A$ by a vector in $\Z_+^n$ models a single step of the $\leb$ process. More precisely, for any collection of triangles $S$, the vector that represents the collection of triangles obtained after applying $\leb$ once to all the triangles in $S$ is $A\cdot v_S$. Furthermore, note that the $\leb$ process partitions all triangles in each step into triangles of the same area. Thus, for every $m\in\N$, the probability vector $\displaystyle{\frac{A^m \cdot e_1}{\norm{A^m \cdot e_1}}}$ represents the distribution of areas occupied by the different classes in $\leb(z)$, after $m$ steps.

\begin{lemma}\label{lem:spectral_radius}
    For every triangle $z$, the matrix $A$ satisfies the following:
    \begin{enumerate}[(a)]
        \item 
        $\lambda = 2$ is an eigenvalue of $A$ and the spectral radius of $A$.
        \item 
        Let $E_2$ and $E_{-2}$ denote the eigenspaces of $\lambda = 2$ and $\lambda=-2$ respectively. Then $\dim\left(E_2\right) = \dim\left(E_{-2}\right) =  q(z)$, the number of terminal quadruples in $\leb(z)$. Furthermore, the eigenvalues $\pm 2$ do not have any non-trivial Jordan blocks. 
    \end{enumerate}
\end{lemma}
\begin{proof}
    \begin{enumerate}[(a)]
        \item 
        Since the column sums of $A$ are all equal to $2$, we have that $\lambda = 2$ is an eigenvalue of $A^t$, with the all-ones vector $\mathds{1}$ as an eigenvector. Since the sets of eigenvalues of $A$ and of $A^t$ are equal, we know that $\lambda = 2$ is an eigenvalue of $A$. By Theorem \ref{thm:bounds_on_rho(A)}, we get that $2$ is the spectral radius of $A$ in this case.

        \item 
        Let $1\leq k=q(z)$ be the number of terminal quadruples in $\gleb(z)$, then we can find an $n\times n$ permutation matrix $P$ such that the matrix $P^tAP$ has the following block triangular form: 
        \[
        P^tAP=
        \left(
        \begin{array}{ccc|cccc} 
            & & & 0 & \cdots & \cdots & 0 \\
            & C & & \vdots & & & \vdots \\
            & & & 0 & \cdots & \cdots & 0 \\
            \hline 
            * & \cdots & * & B_1 & & & \\
            \vdots & & \vdots & & B_2 & & \\
            \vdots & & \vdots & & & \ddots & \\
            * & \cdots & * & & & & B_k \\
        \end{array}
        \right)
        \]
 
        where the first columns of $A$ and of $P^tAP$ are the same, and 
        $
        B_i=
        \begin{pmatrix}
            0 & 0 & 1 & 1 \\
            0 & 0 & 1 & 1 \\
            1 & 1 & 0 & 0 \\
            1 & 1 & 0 & 0
        \end{pmatrix}
        $
        for each $1\leq i\leq k$. Denote by $e_j$ the $j$-th vector of the standard basis of $\R^n$. Thus, for every $1\leq i\leq k$ the vector
        \begin{equation}\label{eq:v_i}
        v_i=e_{n-4(k-i)} + e_{n-4(k-i)-1} + e_{n-4(k-i)-2} + e_{n-4(k-i)-3}
        \end{equation}
        is an eigenvector of $P^tAP$ corresponding to the eigenvalue $\lambda = 2$, and the vectors $v_1,\dots,v_k$ are linearly independent, so we obtain
        \[
        \dim\left(E_2\right)\geq k.
        \]
        In a similar way, for every $1\leq i\leq k$, the vector
        \begin{equation}\label{eq:tilde_v_1}
        \tilde v_i=e_{n-4(k-i)} + e_{n-4(k-i)-1} - e_{n-4(k-i)-2} -e_{n-4(k-i)-3}
        \end{equation}
        is an eigenvector of $P^tAP$ corresponding to the eigenvalue $\lambda = -2$, and the vectors $\tilde v_1,\dots,\tilde v_k$ are linearly independent, so
        \[
        \dim\left(E_{-2}\right)\geq k.
        \]
        For the reverse inequality, we apply Theorem \ref{thm:basis_of_eigenvectors} to $A^2$. 
        The vectors $v_1,\ldots,v_k,\tilde v_1,\ldots,\tilde v_k$ are $2k$ linearly independent eigenvectors of $A^2$ that correspond to the eigenvalue $\lambda=4$, and thus its geometric multiplicity is at least $2k$.  

        By Proposition~\ref{prop:flow_to-A}, from every vertex there is a path to a vertex in a terminal quadruple. In particular, any strongly connected component $\alpha_j$ of $\gleb(z)$ that does not correspond to a terminal quadruple has at least one outgoing edge. This implies that its corresponding adjacency matrix $B_j$ is irreducible (since $\alpha_j$ is strongly connected) and has at least one column sum that is strictly smaller than $2$. Applying Lemma \ref{lem:non-final_irreducible_blocks}, we deduce that any such component $\alpha_j$ has $\rho(\alpha_j)<2$. Thus, the strongly connected components of $\gleb(z)$ with spectral radius $2$ are precisely the terminal quadruples. 
        
        Observe that, for the action of $A^2$, each terminal quadruple is partitioned into $2$ sets of two terminal triangles each, which form all the strongly connected components. Hence, the fact that the dimension of the generalized eigenspace of $\rho(A^2) =4$ is $2k$ follows from Theorem \ref{thm:basis_of_eigenvectors}. In particular, $\dim(E_{2})=\dim(E_{-2})=k$ and the eigenvalues $\pm 2$ do not have any non-trivial Jordan blocks.
    \end{enumerate}
\end{proof}

\begin{lemma}\label{lem:perp_to_1}
Every generalized eigenvector $u$ of $A$ that corresponds to an eigenvalue $\mu\neq 2$ satisfies $u\in \mathds{1}^\perp$.
\end{lemma}

\begin{proof}
    Let $(u_1,\ldots,u_r)$ be a Jordan chain of $A$ that corresponds to an eigenvalue $\mu\neq 2$. That is, $A u_1 = \mu u_1$, $u_1\neq 0$, and for every $2\le i\le r$ we have $A u_{i}=\mu u_i +u_{i-1}$. 

    We first observe that $u_1\in \mathds{1}^\perp$. Indeed,
    \[
    \mu\inpro{u_1}{\mathds{1}} = \inpro{\mu u_1}{\mathds{1}} = \inpro{A u_1}{\mathds{1}} = \inpro{u_1}{A^t\cdot\mathds{1}} = \inpro{u_1}{2\cdot\mathds{1}} = 2\inpro{u_1}{\mathds{1}}. 
    \]
    Since $\mu\neq 2$ we have $\inpro{u_1}{\mathds{1}} = 0$. 

    We proceed by induction. Suppose that $u_i\in \mathds{1}^\perp$. Since $A u_{i+1}=\mu u_{i+1}+u_i$, we have 
    \[2\inpro{u_{i+1}}{\mathds{1}} = \inpro{u_{i+1}}{A^t \mathds{1}} = \inpro{A u_{i+1}}{\mathds{1}} = \inpro{\mu u_{i+1}+u_i}{\mathds{1}} = \mu\inpro{u_{i+1}}{\mathds{1}} + \underbrace{\inpro{u_i}{\mathds{1}}}_{=0}, \] 
    and therefore $\inpro{u_{i+1}}{\mathds{1}} = 0$.
\end{proof}

The proof of Theorem \ref{thm:occupying_area} is implied by the more detailed proposition below. 

\begin{proposition}
Let $z$ be a triangle and $A$ as above. Then
    \begin{enumerate}[(a)]
        \item    
    The limits 
    $\we = \lim_{m\to\infty}\limits\left(\frac{A}{2}\right)^{2m}e_1$ and $\wo = \lim_{m\to\infty}\limits\left(\frac{A}{2}\right)^{2m+1}e_1$ exist. 
    \item 
    The vectors $\we$ and $\wo$ are probability vectors, and there exist $v\in E_2, \tilde v\in E_{-2}$ such that $\we = v+\tilde v$ and $\wo = v - \tilde v$. 
    \item    
    For $j \in \mathbb{N}$ let $w_j\in \R^{l(z)}$ be the probability vector that describes the partition of the area of $z$ into triangles of different similarity classes, after $j$ steps. Then there exists $\xi\in (0,1)$ such that 
    \[\begin{cases}
        \norm{\wo - w_j} = O(\xi^j),& j \text{ is odd. }\\
        \norm{\we - w_j} = O(\xi^j),& j \text{ is even. }
    \end{cases}\]
    \end{enumerate}
\end{proposition}

\begin{proof}
  \begin{enumerate}[(a)]
      \item
        It suffices to show that the first limit exists and then use $\wo=\frac{A}{2}\cdot \we$. 

        By Lemma \ref{lem:spectral_radius}, the spectral radius of $A^2$ is equal to $4$. Let $4=\lambda_1>|\lambda_2|\ge\ldots\ge |\lambda_s|$ be the distinct eigenvalues of $A^2$, in a decreasing order in absolute value. Since the space is a direct sum of the generalized eigenspaces, for each $1\le i\le s$ let $u_i$ be a vector in the generalized eigenspace of $\lambda_i$ such that 
        \begin{equation*}
            e_1 = \sum_{i=1}^su_i.
        \end{equation*}
        Since $e_1\notin \mathds{1}^\perp$, by Lemma \ref{lem:perp_to_1}, $u_1\neq 0$. By Lemma \ref{lem:spectral_radius}, the eigenvalues $2$ and $-2$ of $A$ do not have non-trivial Jordan blocks, so we have
        \[
        \left(\frac{A}{2}\right)^{2m}e_1 = \frac{(A^2)^m\cdot \sum_{i=1}^su_i}{4^m} = u_1+\frac{\sum_{i=2}^s (A^2)^m\cdot u_i}{4^m}, 
        \]
        thus
        \begin{equation}\label{eq:rate_of_convergence}
        \norm{\left(\frac{A}{2}\right)^{2m}e_1 -u_1} \le   \sum_{i=2}^s \frac{|\lambda_i|^m\cdot |g_i(m)|\cdot \norm{u_i}}{4^m}, 
        \end{equation}
        where each $g_i(m)$ is a fixed polynomial whose degree is at most the size of the maximal Jordan block of $\lambda_i$. Since $|\lambda_2|,\ldots,|\lambda_s|<4$, we have $\lim_{m\to \infty}\limits\left(\frac{A}{2}\right)^{2m}e_1 = u_1$. 
        \item 
        Since $\frac{A}{2}$ is a non-negative matrix with all column sums equal to $1$, it preserves the set of non-negative probability vectors. In particular, $\left(\frac{A}{2}\right)^j e_1$ is a probability vector for every $j\in\N$, and hence so are $\we$ and $\wo$.
        Note that the eigenspace of $A^2$ corresponding to $\lambda = 4$ is $E_2\oplus E_{-2}$. Since $\we = u_1$, there exist $v\in E_2$, $\tilde v\in E_{-2}$ such that $\we = v+\tilde v$. Furthermore, the equality $\wo = v - \tilde v$ follows from the fact that $\wo=\frac{A}{2}\cdot \we$. 
      \item 
        From the discussion above, it is clear that $w_j = \left(\frac{A}{2}\right)^j e_1$. Then for $j=2m$, we have 
        \[
        \norm{\we - w_j} \stackrel{\eqref{eq:rate_of_convergence}}= \sum_{i=2}^s \frac{|\lambda_i|^m\cdot |g_i(m)|\cdot \norm{u_i}}{4^m} = O(\xi^j), 
        \]
        where $\xi$ is any number in the interval $\left(\frac{\sqrt{|\lambda_2|}}{2},1\right)$. The odd case is proven similarly, and this completes the proof. \qedhere
  \end{enumerate}     
\end{proof}

\section{Concluding remarks and open problems}\label{sec:concluding}

\subsection{A remark on the geometric distribution of $\leb(z)$}\label{subsec:leb_geo_dist}
As we already mentioned in~\S\ref{subsec:geometric_structure}, we have observed several surprising properties of the geometric position of the points of $\leb(z)$ in $\cD$, in particular, Proposition~\ref{prop:terminal_sets_are_on_a_circle} characterized the geometric position of a terminal quadruple. 

Furthermore, as we have seen in the proof of Lemma~\ref{lem:2^m_points_in_B}, the points 
\[\set{ h^j(z_m)}{j\in \{0,\ldots,2^{m-1}\} }\subseteq \V\] 
lie on the circle $C_m$, which is tangent to the real axis at the origin. As in the proof of Theorem~\ref{thm:arbitrarily_many_quadruples}, applying $L^m$ to this set produces $2^{m-1}+1$ distinct triangles in region $\I$ that lie on a circle that is tangent to the real axis at the origin. In fact, these $2^{m-1}+1$ points belong to distinct terminal quadruples: since $L$ in region $\I$ is an inversion with respect to the geodesic $\absolute{z}=\frac{\sqrt2}{2}$, for every $z\in \I$, the only circle that is centered at the imaginary axis and passes through the points $z$ and $L(z)$ is the one centered at $\frac{\sqrt{2}}{2}i$. A direct computation shows that this circle is not tangent to the real axis at $0$, since the circle $C_z$ whose center is at $\frac{\sqrt{2}}{2}i$ does not intersect region $\I$.

We conjecture that the whole orbit $\leb(z)$ can be partitioned into a certain number of subsets, each lying on a certain arc that is tangent to the real axis at $0$ or at $\frac{1}{2}$. The precise number and location of these arcs, and the precise number and location of points on these arcs requires a more refined analysis in the spirit that we have performed. We leave this for future research.

\subsection{A remark on continuity and small perturbations}
Since a small perturbation of the vertices of a triangle produces a nearby triangle, it is natural to consider the relation between $\leb(z)$ and $\leb(z')$ for close triangles $z$ and $z'$. Let $d:\HH\times\HH\to [0,\infty)$ be the hyperbolic metric. For a set $F\subseteq\HH$ and $\varepsilon>0$ we denote by 
    \[
    N_\varepsilon[F] = \bigcup_{x\in F}B(x,\varepsilon),
    \]
the $\varepsilon$-neighborhood of $F$. 
In \cite[Lemma 3.2]{PerdP14}, the following non-increasing property of the metric $d$ was proven. For every $z,w\in \cD$ we have 
    \[
    d(L(w),L(z))\le d(w,z) \quad\text{ and }\quad d(R(w),R(z))\le d(w,z).
    \]
    
We mention the following continuity result of $\leb$, which is an immediate corollary of the non-increasing property.

\begin{corollary}\label{cor:perturbation}
    For every $z\in \cD$ we have the following: for every $\varepsilon>0$ and every $w\in \cD$ with $d(z,w)<\varepsilon$ we have $\leb(w) \subseteq N_\varepsilon[\leb(z)]$. 
\end{corollary}
This phenomenon is exemplified in Figure \ref{fig:perturbation} below, where one can readily see that triangles with low complexity  can have nearby triangles with high complexity.
\begin{figure}[h!]
    \begin{tabular}{@{\hspace{-2.3em}}c@{\hspace{-2.3em}}c}
    \includegraphics[width=0.59\linewidth]{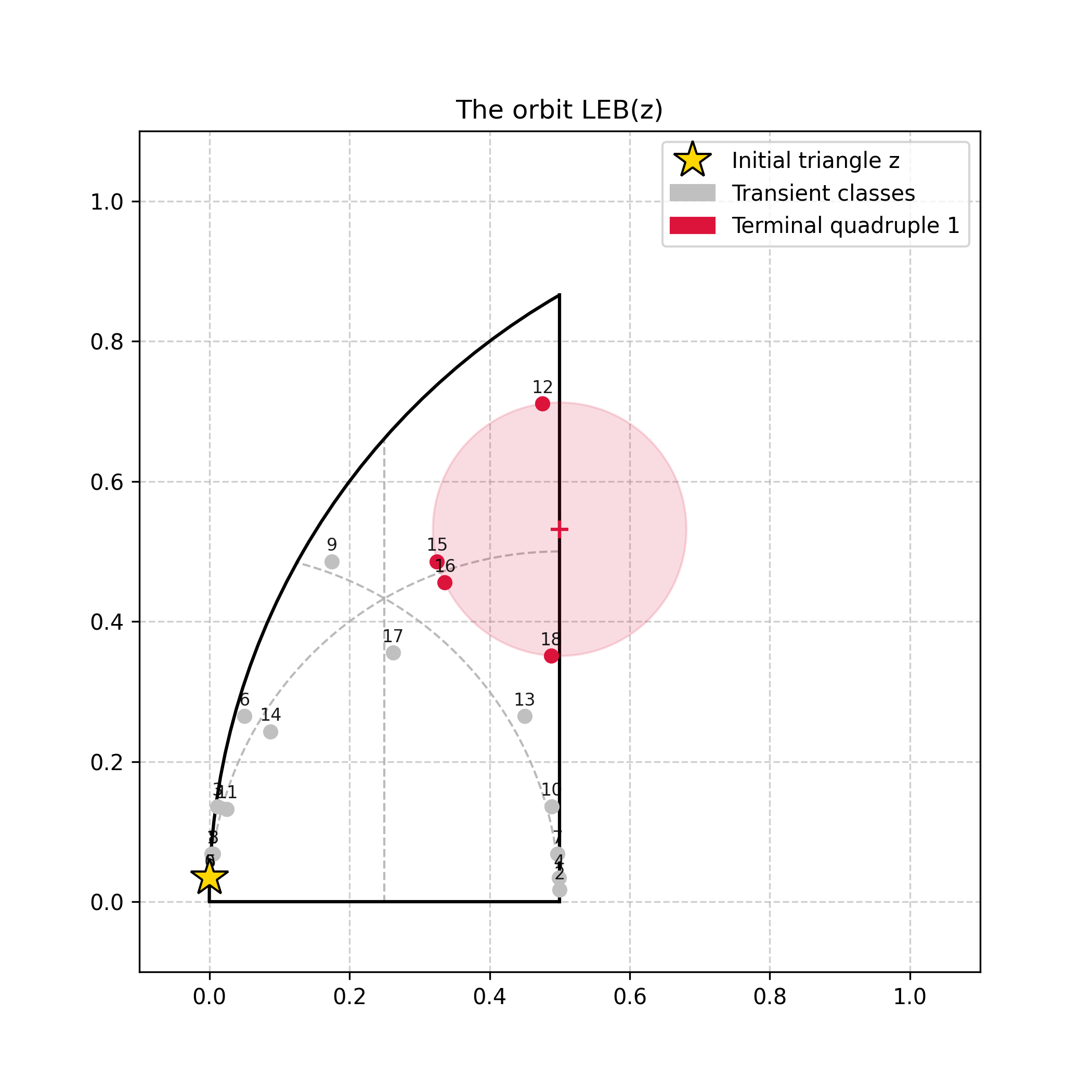}
    &
    \includegraphics[width=0.59\linewidth]{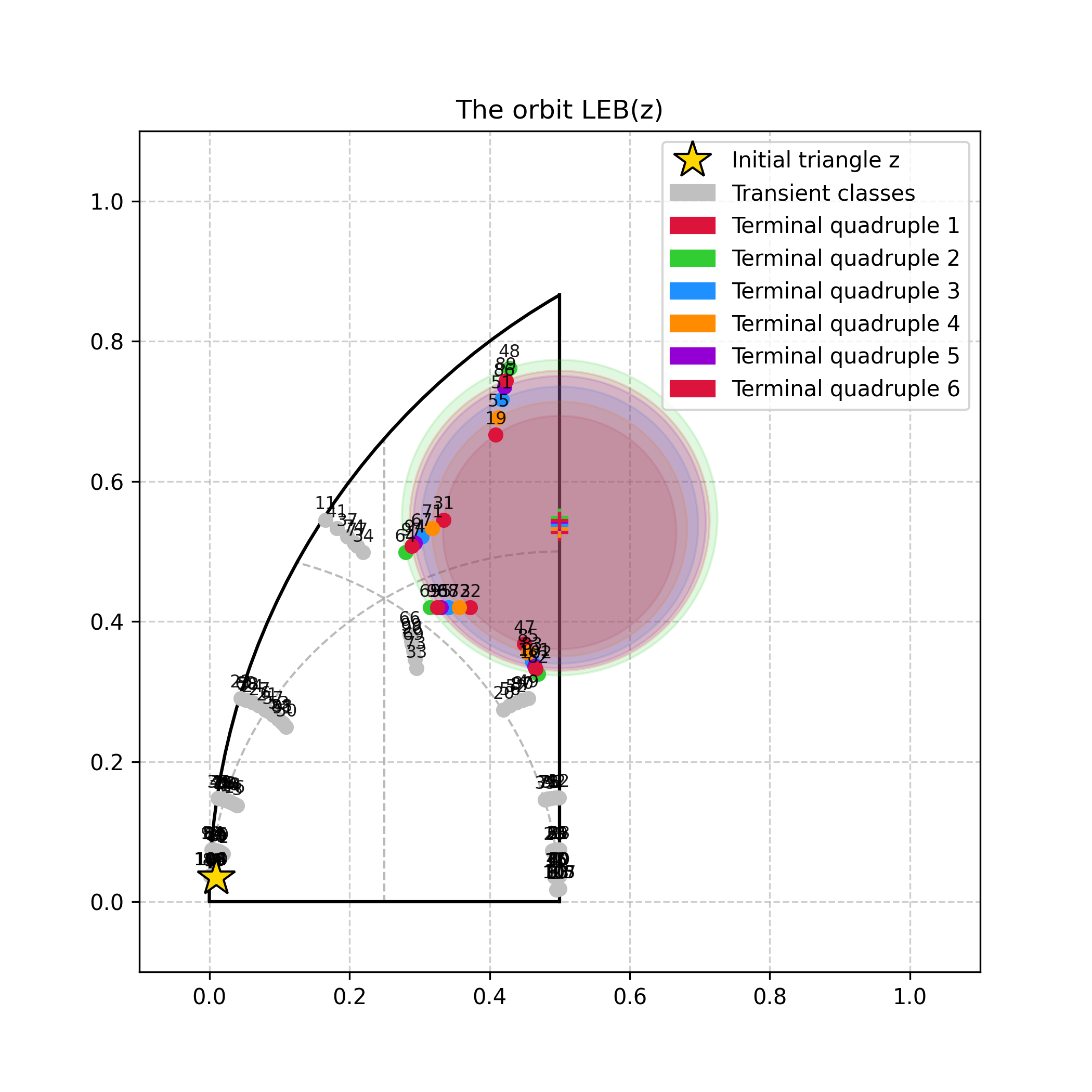}
    \end{tabular}
    \caption{On the left, we see the orbit of the isoceles triangle $z=\frac{1}{1700}+\frac{\sqrt{3399}}{1700}i$, which has a total of $19$ similarity classes and $1$ terminal quadruple. On the right, we see the orbit of the nearby triangle $z'=\frac{1}{100}+\frac{\sqrt{3399}}{1700}i$, which has a total of $109$ similarity classes and $6$ terminal quadruples.}
    \label{fig:perturbation}
\end{figure}

\subsection{Open problems}
In view of computational observations about $\gleb(z)$ we ask the following.
\begin{enumerate}
    \item 
    Is it true that the spectrum of $\gleb(z)$ is contained in $\{\pm 2, \pm\sqrt{2},\pm 1, 0\}$?
    \item Furthermore, is it true that $\gleb(z)$ is bipartite?
    \item 
    Classify the cases where the adjacency matrix of $\gleb(z)$ is diagonalizable.
\end{enumerate}

In view of Theorem~\ref{thm:occupying_area}, an important complexity measure of the mesh is the number $q(z)$ of terminal quadruples, thus we propose:
\begin{enumerate}
    \setcounter{enumi}{3}
    \item Find a formula for $q(z)$ in terms of $z$. 
\end{enumerate}

\subsection*{Acknowledgements} We thank Nati Linial for valuable comments. The first author thanks Karim Adiprasito for introducing him to the $\leb$ process.


\bibliographystyle{alpha}
\bibliography{leb}

\end{document}